\newif\ifstoc
\stocfalse

\ifstoc
  \documentclass{sig-alternate}
  \usepackage{times}
\else
  \documentclass[letterpaper,11pt]{article}

  \usepackage{typearea}
  \typearea{15}

  \usepackage[compact]{titlesec}

\fi

\newcommand{\stocoption}[2]
{\ifstoc%
#1%
\else%
#2%
\fi}

\usepackage{theorem,latexsym,graphicx}
\usepackage{amsmath,amssymb,enumerate}
\usepackage{xspace}%\includegraphics[]{sp-cut.pdf}

\usepackage{bm}
\usepackage{ifpdf}
\usepackage{shadow,shadethm,color}
\usepackage{algorithm}
\usepackage[noend]{algorithmic}
\usepackage{subfigure}
\usepackage{verbatim}
\usepackage{paralist}

\allowdisplaybreaks

\ifstoc
\newcommand{\lref}[2][]{{#1~\ref{#2}}}
\else
\definecolor{Darkblue}{rgb}{0,0,0.4}
\definecolor{Brown}{cmyk}{0,0.81,1.,0.60}
\definecolor{Purple}{cmyk}{0.45,0.86,0,0}
\newcommand{\mydriver}{hypertex}
\ifpdf
 \renewcommand{\mydriver}{pdftex}
\fi
\usepackage[breaklinks,\mydriver]{hyperref}
\hypersetup{colorlinks=true,%pdfborder={1 1 1 [3]},%
            citebordercolor={.6 .6 .6},linkbordercolor={.6 .6 .6},%
citecolor=Darkblue,urlcolor=black,linkcolor=Darkblue,pagecolor=black}
\newcommand{\lref}[2][]{\hyperref[#2]{#1~\ref*{#2}}}
\fi

\ifstoc
\else
\fi

\newtheorem{theorem}{Theorem}[section]

\newtheorem{conjecture}[theorem]{Conjecture}
\newtheorem{lemma}[theorem]{Lemma}
\newshadetheorem{lemmashaded}[theorem]{Lemma}

\newtheorem{observation}[theorem]{Observation}
\newtheorem{claim}[theorem]{Claim}

\newtheorem{corollary}[theorem]{Corollary}

\numberwithin{algorithm}{section}

\ifstoc
\else
\newenvironment{proof}{

\noindent{\bf Proof:}}
{\hfill$\blacksquare$

}
\fi

\newcommand{\junk}[1]{}
\newcommand{\ignore}[1]{}

\newcommand{\R}[0]{{\ensuremath{\mathbb{R}}}}

\newcommand{\Z}[0]{{\ensuremath{\mathbb{Z}}}}

\newcommand{\poly}{\operatorname{poly}}

\newcommand{\sse}{\subseteq}

\newcommand{\T}{{\mathcal{T}}}

\newcommand{\e}{\varepsilon}
\newcommand{\eps}{\varepsilon}
\newcommand{\ts}{\textstyle}

\newcommand{\ALG}{\ensuremath{{\sf alg}}}

\newcommand{\LPOpt}{\ensuremath{\mathsf{LP^\star}\xspace}}

\newcounter{note}[section]

\newcommand{\qedsymb}{\hfill{\rule{2mm}{2mm}}}
\renewenvironment{proof}{\begin{trivlist} \item[\hspace{\labelsep}{\bf
\noindent Proof.\/}] }{\qedsymb\end{trivlist}}%

\newcommand{\initOneLiners}{%
    \setlength{\itemsep}{0pt}
    \setlength{\parsep }{0pt}
    \setlength{\topsep }{0pt}
}
\newenvironment{OneLiners}[1][\ensuremath{\bullet}]
    {\begin{list}
        {#1}
        {\initOneLiners}}
    {\end{list}}

\newcommand{\squishlist}{
 \begin{list}{$\bullet$}
  { \setlength{\itemsep}{0pt}
     \setlength{\parsep}{3pt}
     \setlength{\topsep}{3pt}
     \setlength{\partopsep}{0pt}
     \setlength{\leftmargin}{1.5em}
     \setlength{\labelwidth}{1em}
     \setlength{\labelsep}{0.5em} } }

\newcommand{\squishend}{
  \end{list}  }

\newcommand{\E}{\mathbb{E}}
\newcommand{\D}{\mathcal{D}}
\newcommand{\F}{\mathcal{F}}
\newcommand{\jay}{D}
\newcommand{\maxcut}{\textsf{mc}}
\newcommand{\Abar}{\ensuremath{\overline{A}}}

\newcommand{\mc}{\textsc{MaxCut}\xspace}
\newcommand{\spcut}{Sparsest Cut\xspace}
\newcommand{\nusc}{Non-Uniform Sparsest Cut\xspace}
\newcommand{\usc}{Uniform Sparsest Cut\xspace}
\newcommand{\SA}{Sherali-Adams\xspace}
\newcommand{\pick}{\text{pick }}
\newcommand{\capa}{\textsf{cap}}
\newcommand{\dem}{\textsf{dem}}
\newcommand{\lca}{\textsf{lca}}

\newcommand{\ess}{\mathcal{S}}
\newcommand{\bee}{\mathcal{B}}
\newcommand{\dee}{\mathcal{D}}

\newcommand{\ulc}{Unique Label Cover\xspace}
\newcommand{\awesome}{nice\xspace}
\newcommand{\bits}{\ensuremath{\{-1,1\}}}
\newcommand{\Qhat}{\ensuremath{\widehat{Q}}}
\newcommand{\xhat}{\ensuremath{\widehat{x}}}

\newcommand{\cut}[1]{(#1,\overline{#1})}

\begin{document}

\ifstoc
\conferenceinfo{STOC'13,} {June 1-4, 2013, Palo Alto, California, USA.}
\CopyrightYear{2013}
\crdata{978-1-4503-2029-0/13/06}
\clubpenalty=10000
\widowpenalty = 10000
\fi

\title{Sparsest Cut on Bounded Treewidth Graphs: \\ Algorithms and
  Hardness Results}
\author{
Anupam Gupta\thanks{Department of Computer Science, Carnegie Mellon
  University, Pittsburgh PA 15213, and Department of IEOR, Columbia University. Research was partly supported by
    NSF awards CCF-0964474 and CCF-1016799, and by a grant from the
    CMU-Microsoft Center for Computational Thinking.}
\and
Kunal Talwar\thanks{Microsoft Research SVC, Mountain View, CA 94043.}
\and
David Witmer\thanks{Department of Computer Science, Carnegie Mellon
  University, Pittsburgh PA 15213. Research was partly supported by
    NSF awards CCF-1016799, CCF-0747250, and CCF-0915893, and by a Sloan fellowship. }
}
\date{}

\maketitle

\begin{abstract}
  We give a $2$-approximation algorithm for \nusc
  that runs in time $n^{O(k)}$, where $k$ is the treewidth of the
  graph. This improves on the previous $2^{2^k}$-approximation in time \ifstoc
  $\poly(n)$ $2^{O(k)}$ \else $\poly(n) 2^{O(k)}$ \fi due to Chlamt\'a\v{c} et al.~\cite{CKR10}.

  To complement this algorithm, we show the following hardness results:
  If the \nusc problem has a
  $\rho$-approx\-imation for series-parallel graphs (where $\rho \geq 1$),
  then the \mc problem has an algorithm with approximation factor
  arbitrarily close to $1/\rho$. Hence, even for such restricted graphs
  (which have treewidth $2$), the Sparsest Cut problem is
  \textsc{NP}-hard to approximate better than $17/16 - \eps$ for $\eps > 0$; assuming
  the Unique Games Conjecture the hardness becomes $1/\alpha_{GW} - \eps$.
  For graphs with large (but constant) treewidth, we show a hardness
  result of $2 - \eps$ assuming the Unique Games Conjecture.

  Our algorithm rounds a linear program based on (a subset
  of) the \SA lift of the standard Sparsest Cut LP. We show that even
  for treewidth-$2$ graphs, the LP has an integrality gap close to $2$
  even after polynomially many rounds of \SA. Hence our approach cannot
  be improved even on such restricted graphs without using a stronger
  relaxation.
\end{abstract}

\ifstoc
%A category including the fourth, optional field follows...
\fi

\section{Introduction}

The \emph{\spcut} problem takes as input a ``supply'' graph $G =
(V,E_G)$ with positive edge capacities $\{\capa_e\}_{e \in E_G}$, and a ``demand''
graph $\jay = (V,E_\jay)$ (on the same set of vertices $V$) with demand values
$\{\dem_e\}_{e \in E_\jay}$, and aims to determine
\[ \Phi_{G,\jay} := \min_{S \sse V} \frac{ \sum_{e \in \partial_G(S)}
  \capa_e } { \sum_{e \in \partial_\jay(S)} \dem_e }, \] where
$\partial_G(S)$ denotes the edges crossing the cut $(S, V \setminus S)$
in graph $G$. When $E_\jay = \binom{V}{2}$ with $\dem_e = 1$, the
problem is called \emph{Uniform Demands} Sparsest Cut, or simply
\emph{\usc.} Our results all hold for the non-uniform demands case.

The \spcut problem is known to be NP-hard due to a result of Matula and
Shahrokhi~\cite{MS90}, even for unit capacity edges and uniform demands.
 The best algorithm for \usc on general graphs
is an $O(\sqrt{\log n})$-approximation due to Arora, Rao, and
Vazirani~\cite{ARV04}; for \nusc the best factor is $O(\sqrt{\log n}
\log\log n)$ due to Arora, Lee and Naor~\cite{ALN05}. An older
$O(\sqrt{\log n})$-approximation for \nusc is known for all
exclud\-ed-minor families of graphs~\cite{Rao99}, and constant-factor
approximations exist for more restricted classes of graphs~\cite{GNRS99,
  CGNRS01, CJLV08, LR10, LS09, CSW10}. Constant-factor approximations
are known for \usc for all exclud\-ed-minor families of
graphs~\cite{KPR93, Yuri03}. \cite{GS12} give a $(1+\e)$-approx\-imation
algorithm for non-uniform Sparsest Cut that runs in time depending on
generalized spectrum of the graphs $(G, \jay)$. All above results,
except~\cite{GS12}, consider either the standard linear or
SDP relaxations. The integrality gaps of convex relaxations of
Sparsest Cut are intimately related to questions of embeddability of
finite metric spaces into $\ell_1$; see, e.g.,~\cite{LLR95, GNRS99, KV05, KR09, LeeNaor06, CKN09, LS11,
  CKN11} and the many references therein. 
Integrality gaps for LPs/SDPs obtained from lift-and-project techniques
appear in \cite{CMM09, KS09, RS09, GSZ12}.
\cite{GNRS99}
conjectured that metrics supported on graphs excluding a fixed minor
embed into $\ell_1$ with distortion~$O(1)$ (depending on the excluded
minor, but independent of the graph size); this would imply
$O(1)$-approximations to \nusc on instances $(G,\jay)$ where $G$
excludes a fixed minor. This conjecture has been verified for several
classes of graphs, but remains open (see, e.g.,~\cite{LS09} and
references therein).

The starting point of this work is the paper of Chlamt\'a\v{c} et
al.~\cite{CKR10}, who consider non-uniform
\spcut on graphs of treewidth $k$.\footnote{We emphasize that only the
  supply graph $G$ has bounded treewidth; the demand graphs $\jay$ are
  unrestricted.}  They ask if one can obtain good algorithms for such
graphs without answering the~\cite{GNRS99} conjecture; in particular, they
look at the \SA hierarchy. In their paper, they give an
$2^{2^k}$-approximation in time $\poly(n)\,2^{O(k)}$ by solving the
$k$-round \SA linear program and ask whether one can achieve an
algorithm whose approximation ratio is independent of the treewidth $k$.
We answer this question in the affirmative.

\begin{theorem}[Easiness]
  \label{thm:main2}
  There is an algorithm for the \nusc problem that, given any instance
  $(G,\jay)$ where $G$ has treewidth $k$, outputs a $2$-approximation  
  in time $n^{O(k)}$.
\end{theorem}

Graphs that exclude some planar graph as a minor have bounded treewidth,
and $H$-minor-free graphs have treewidth $O(|H|^{3/2} \sqrt{n})$. This
implies a $2$-approximation for planar-minor-free graphs in poly-time,
and for general minor-free graphs in time $2^{O(\sqrt{n})}$. In
fact, we only need $G$ has a recursive vertex separator decomposition
where each separator has $k$ vertices for the above theorem to apply.

Our algorithm is also based on solving an LP relaxation, one whose
constraints form a subset of the $O(k \log n)$-round \SA lift of the
standard LP, and then rounding it via a natural propagation rounding
procedure. We show \stocoption{in the full version}{} that further applications of the \SA operator (even
for a polynomial number of rounds) cannot do better:
\begin{theorem}[Tight Integrality Gap]
  \label{thm:main3}
  For every $\eps > 0$, there are instances $(G,\jay)$ of the \nusc problem
  with $G$ having treewidth 2 (a.k.a.\ series-parallel graphs) for which the integrality gap after
  applying $r$ rounds of the \SA hierarchy still remains $2 -
  \eps$, even when $r = n^{\delta}$ for some constant $\delta =
  \delta(\eps) > 0$.
\end{theorem}
This result extends the integrality gap lower bound for the basic LP on
series-parallel graphs shown by Lee and
Raghavendra~\cite{LR10}, for which Chekuri, Shepherd and Weibel gave a different proof~\cite{CSW10}.

On the hardness side, Amb\"{u}hl et al.~\cite{AMS07} showed that if \usc
admits a PTAS, then SAT has a randomized \ifstoc sub-expo-nential \else sub-exponential \fi time algorithm.
Chawla et al.~\cite{CKKRS05} and Khot and Vishnoi~\cite{KV05} showed
that \nusc is hard to approximate to any constant factor, assuming the
Unique Games Conjecture.  The only \textsc{Apx}-hardness result 
%(i.e., constant-factor hardness-of-approximation assuming 
(based on $P \neq NP$) for
\nusc is recent, due to Chuzhoy and Khanna~\cite[Theorem~1.4]{CK06}.  Their
reduction from \mc shows that the problem is \textsc{Apx}-hard even when $G$ is
$K_{2,n}$, and hence of treewidth or even pathwidth 2. (This reduction was
rediscovered by Chlamt\'a\v{c}, Krauthgamer, and Raghavendra~\cite{CKR10}.) We
extend their reduction to show the following hardness result
for the \nusc problem:
\begin{theorem}[Improved NP-Hardness]
  \label{thm:main1}
  For every constant $\eps > 0$, the \nusc problem is hard to
  approximate better than $\frac{17}{16} - \eps$ unless $P=NP$ and hard to approximate 
better than $1/\alpha_{GW} - \eps$ assuming the Unique Games
  Conjecture, even on graphs with treewidth 2
  (series-parallel graphs). 
\end{theorem}
Our proof of this result gives us a hardness-of-approximation that is
essentially the same as that for \mc (up to an additive $\eps$ loss).
Hence, improvements in the NP-hardness for \mc would translate into
better NP-hardness for \nusc as well.

If we allow instances of larger treewidth, we get a 
Unique Games-based hardness that matches our algorithmic guarantee:
\begin{theorem}[Tight UG Hardness]
  \label{thm:main4}
  For every constant $\eps > 0$, it is UG-hard to approximate \nusc on
  bounded treewidth graphs better than $2-\eps$. I.e., the
  existence of a family of algorithms, one for each treewidth $k$, that
  run in time $n^{f(k)}$ and give $(2-\eps)$-approximations for \nusc would disprove the Unique Games Conjecture.
\end{theorem}
%  for every eps, there is some eta (much smaller than eps), such
% that 2-eps for SC on TW(2^d) graphs gives us 1-eta vs eta for UG.
% And for eps constant, eta and hence d are both constants, so this
% would be a poly-time algorithm for UG(eta). 

% We leave open the question of whether using a stronger hierarchy (such
% as the Lasserre hierarchy) would give an algorithm whose approximation
% guarantee is closer to $\approx 1/\alpha_{GW}$ for bounded treewidth
% graphs. % We are able to get an approximation of $\approx 1/0.85$ for the
% kind of treewidth-2 instances $(G,D)$ in our hardness result (by
% combining the rounding algorithm of~\cite{RT12} with our rounding
% procedure), but the problem for

\subsection{Other Related Work}
\label{sec:other-related-work}

There is much work on algorithms for bounded treewidth graphs: many
NP-hard problems can be solved exactly on such graphs in polynomial time
(see, e.g.,~\cite{RS-II}).  Bienstock and Ozbay~\cite{BO04} show, e.g.,
that the stable set polytope on treewidth-$k$ graphs is integral after
$k$ levels of \SA; Magen and Moharrami~\cite{MM09} use their result to
show that $O(1/\eps)$ rounds of \SA are enough to $(1+\eps)$-approximate
stable set and vertex cover on minor-free graphs.
Wainwright and Jordan~\cite{WJ04} show conditions under which \SA and
Lasserre relaxations are integral for combinatorial problems based on
the treewidth of certain hypergraphs.  In contrast, our lower bounds
show that the \spcut problem is \textsc{Apx}-hard even on treewidth-2
supply graphs, and the integrality gap stays close to $2$ even after a
polynomial number of rounds of \SA.

\section{Preliminaries and Notation}
\label{sec:prelims}

We use $[n]$ to denote the set $\{1, 2, \ldots, n\}$. For a set $A$ and
element $i$, we use $A+i$ to denote $A \cup \{i\}$.

\subsection{Cuts and \mc Problem}
All the graphs we consider are undirected. For a graph $G = (V,E)$ and
set $S \sse V$, let $\partial_G(S)$ be the edges with exactly one
endpoint in $S$; we drop the subscript when $G$ is clear from context.
Given vertices $V$ and special vertices $s, t$, a cut $(A, V \setminus
A)$ is $s$-$t$-separating if $|A \cap \{s,t\}| = 1$. 

In the (unweighted) \mc problem, we are given a graph $G = (V, E)$ and
want to find a set $S \sse V$ that maximizes $|\partial_G(S)|$; the
weighted version has weights on edges and seeks to maximize the weight
on the crossing edges. The approximability of weighted
and unweighted versions of \mc differ only by an $(1 + o(1))$-factor~\cite{CST96}, and henceforth we only consider the unweighted
case.

\subsection{Tree Decompositions and Treewidth}
Given a graph $G = (V, E_G)$, a \emph{tree decomposition} consists of a
tree $\T = (X, E_X)$ and a collection of node subsets $\{U_i \subseteq
V\}_{i \in X}$ called ``bags'' such that the bags containing any node
$v \in V$ form a connected component in $T$ and each edge in $E_G$ lies
within some bag in the collection. The width of such a tree
decomposition is $\max_{i \in X} (|U_i| - 1)$, and the treewidth of $G$
is the smallest width of any tree-decomposition for $G$. See,
e.g.,~\cite{diestel, Bod98} for more details and references.

The notion of treewidth is intimately connected to the underlying graph
$G$ having small vertex separators. Indeed, say graph $G = (V,E)$
\emph{admits (weighted) vertex separators of size $K$} if for every
assignment of positive weights to the vertices $V$, there is a set $X
\sse V$ of size at most $K$ such that no component of $G - X$ contains
more than $\frac23$ of the total weight $\sum_{v \in V} w_v$. For example,
planar graphs admit weighted vertex separators of size at most
$\sqrt{n}$. It is known (see, e.g., \cite[Theorem~1]{Reed92}) that if $G$
has treewidth $k$ then $G$ admits weighted vertex separators of size at
most $k+1$; conversely, if $G$ admits weighted vertex separators of size
at most $K$ then $G$ has treewidth at most $4K$. (The former statement
is easy. A easy weaker version of the latter implication with treewidth
$O(K \log n)$ is obtained as follows. Find an unweighted vertex
separator $X \subseteq V$ of $G$ of size $K$ to get subgraphs $G_1, G_2,
\ldots, G_t$ each with at most $2/3$ of the nodes. Recurse on the
subgraphs $G_i \cup X$ to get decomposition trees $\T_1, \ldots,
\T_t$. Attach a new empty bag $U$ and connecting $U$ to the ``root'' bag
in each $\T_i$ to get the decomposition tree $\T$, add the vertices of
$X$ to all the bags in $\T$, and designate $U$ as its root. Note that
$\T$ has height $O(\log n)$ and width $O(K \log n)$. In fact, this tree
decomposition can be used instead of the one from Theorem~\ref{thm:bodl}
for our algorithm in Section~\ref{sec:approx} to get the same asymptotic
guarantees.) 

\subsection{The \SA Operator}
For a graph with $|V| = n$, we now define the \SA polytope.  We can
strengthen an LP by adding all variables $x(S,T)$ such that $|S| \leq
r$ and $T \sse S$. The variable $x(S,T)$ has the ``intended solution''
that the chosen cut $(A, \Abar)$ satisfies $A \cap S = T$. \footnote{In
  some uses of \SA, variables $x_{S,T}$ are intended to mean that
  $A \cap (S \cup T) = S$---this is not the case here.}
%These $x(S,T)$'s satisfy some ``consistency constraints'' defined below.
We can then define the $r$-round \SA polytope (starting with the trivial
LP), denoted $\mathbf{SA}_r(n)$, to be the set of all vectors
$(y_{uv})_{u,v \in V} \in \mathbb{R}^{\binom{n}{2}}$ satisfying the
following constraints: \ifstoc
\begin{flalign}
&y_{uv} = x(\{u,v\}, \{u\}) + x(\{u,v\}, \{v\})  \quad \forall u,v \in V \label{eqn:sa_edges} \\
&\sum_{T \sse S} x(S,T) = 1 \quad \forall S \sse V \text{ s.t. } |S| \leq r \label{eqn:sa_sum_to_1} \\
&x(S,T) = x(S+u, T) + x(S+u, T+u) \label{eqn:sa_cons}  \\
& \qquad \qquad \quad \forall S \sse V \text{ s.t. } |S| \leq r-1, T \sse S, u \notin S \nonumber \\
&x(S,T) \geq 0 \quad \forall S \sse V \text{ s.t. } |S| \leq r, T \sse S \label{eqn:sa_non_neg}
\end{flalign}
\else
\begin{eqnarray}
y_{uv} &=& x(\{u,v\}, \{u\}) + x(\{u,v\}, \{v\}) \quad \forall u,v \in V  \label{eqn:sa_edges} \\
\sum_{T \sse S} x(S,T) &=& 1 \quad \forall S \sse V \text{ s.t. } |S| \leq r \label{eqn:sa_sum_to_1}\\
x(S,T) &= &x(S+u, T) + x(S+u, T +u) \quad \forall S \sse V \text{ s.t. } |S| \leq r-1, T \sse S, u \notin S \label{eqn:sa_cons}\\
x(S,T) &\geq& 0 \quad \forall S \sse V \text{ s.t. } |S| \leq r, T \sse S \label{eqn:sa_non_neg}
\end{eqnarray}
\fi 

We will refer to \eqref{eqn:sa_cons} as consistency constraints.  These
constraints immediately imply that the $x(S,T)$ variables satisfy the
following useful property:

% This formulation is equivalent to other, more standard formulations of
% the $r$-round \SA polytope, such as the one given in~\cite{CKR10}.  In
% \lref[Section]{sec:same_as_CKR}, we show that $\mathbf{SA}_r(n)$ is
% equivalent to the polytope defined in~\cite{CKR10}.

\begin{lemma}
\label{lem:cons}
For every pair of disjoint sets $S,S' \sse V$ such that $|S \cup S'| \leq r$ and for any $T \subseteq S$, we have:
\begin{equation*}
x(S,T) = \sum_{T' \subseteq S'} x(S \cup S', T \cup T')
\end{equation*}
\end{lemma}
\begin{proof}
This follows by repeated use of \eqref{eqn:sa_cons}.
\end{proof}

We can now use $\mathbf{SA}_r(n)$ to write an LP relaxation for an instance $G=(V,E)$ of \mc:
\begin{equation}
\label{eqn:sa_mc}
\begin{split}
\max & \sum_{(u,v) \in E} y_{uv} \\
\text{s.t.} & \quad y_{uv} \in \mathbf{SA}_r(n) \quad \forall u,v \in V
\end{split}
\end{equation}
We can also define an LP relaxation for an instance $(G,D)$ of \nusc:
\begin{equation}
\label{eqn:sa_sc}
\begin{split}
\min & ~ \frac{\sum_{(u,v) \in E_G} \capa_{uv} y_{uv}}{\sum_{(u,v) \in E_D}
  \dem_{uv} y_{uv}} \\
\text{s.t.} & \quad y_{uv} \in \mathbf{SA}_r(n) \quad \forall u,v \in V
\end{split}
\end{equation}
Note that the Sparsest Cut objective function is a ratio, so this is not
actually an LP as stated.  Instead, we could add the constraint
$\sum_{(u,v) \in E_D} \dem_{uv} y_{uv} \geq \alpha$, minimize $\sum_{(u,v) \in
  E_G} \capa_{uv} y_{uv}$, and use binary search to find the correct value of
$\alpha$.  In \lref[Section]{sec:approx}, we will use (a slight weakening
of) this relaxation in our approximation algorithm for Sparsest Cut on
bounded-treewidth graphs, and in \stocoption{the full version}{\lref[Section]{sec:SA-gaps}} we will show
that \SA integrality gaps for the \mc LP~(\ref{eqn:sa_mc}) can be translated
into integrality gaps for the \spcut LP~(\ref{eqn:sa_sc}).

%%% Local Variables: 
%%% mode: latex
%%% TeX-master: "sp-cut"
%%% End: 

\ifstoc
\section{An Algorithm for Bounded Tree-width Graphs}
\else
\section{An Algorithm for Bounded Treewidth Graphs}
\fi
\label{sec:approx}

In this section, we present a $2$-approximation algorithm for \spcut
that runs in time $n^{O(\text{treewidth})}$. Consider an instance $(G,
D)$ of \spcut, where $G$ has treewidth $k'$, but there are no
constraints on the demand graph $D$. We assume that we are also given an
initial tree-decomposition $(\T' = (X', E_{X'}); \{U'_i \subseteq V \mid
i \in X'\})$ for $G$. This is without loss of generality, since such an
tree-decomposition $T'$ can be found, e.g., in time
$O(n^{k'+2})$~\cite{ACP87} or time $O(n) \cdot
\exp(\poly(k'))$~\cite{Bod96}; a tree-decomposition of
width $O(k' \log k')$ can be found in $\poly(n)$ time~\cite{Amir10}.

\subsection{Balanced Tree Decompositions and the Linear Program}

We start with a result of Bodlaender~\cite[Theorem~4.2]{Bod89} which
converts the initial tree decomposition into a ``nice'' one, while
increasing the width only by a constant factor:
\begin{theorem}[Balanced Tree Decomp.]
  \label{thm:bodl}
  Given graph $G = (V,E_G)$ and a tree decomposition $(\T' = (X',
  E_{X'}); \{U'_i \subseteq V \mid i \in X'\})$ for $G$ with width at
  most $k'$, there is a tree decomposition $(\T = (X, E_X); \{U_i
  \subseteq V \mid i \in X\})$ for $G$ such that
  \begin{OneLiners}
  \item[(a)] $\T$ is a binary tree of depth at most $\lambda := 2 \lceil
    \log_{5/4} (2n) \rceil$, and
  \item[(b)] $\max_{i \in X} |U_i|$ is at most $k := 3k'+3$, and hence
    the width is at most $k-1$.
  \end{OneLiners}
  Moreover, given $G$ and $\T'$, such a decomposition $\T$ can be found in time
  $O(n)$.
\end{theorem}

From this point on, we will work with the balanced
\ifstoc tree de-composition \else tree decomposition \fi $\T = (X, E_X)$, whose root node is denoted by
$r \in X$. Let $P_{ra}$ denote the set of nodes on the tree path in $\T$
between nodes $a,r \in X$ (inclusive), and let $V_a = \cup_{b \in
  P_{ra}} U_b$ be the union of the bags $U_b$'s along this $r$-$a$ tree
path. Note that $|V_a| \leq k\cdot \lambda$.
%For convenience, let us fix an arbitrary node in the root bag $z \in U_r$.

%\subsection{The Linear Program}

Recall the \SA linear program \eqref{eqn:sa_sc}, with variables $x(S,T)$
for $T \subseteq S$ having the intended meaning that the chosen cut $(A,
\Abar)$ satisfies $A \cap S = T$. We want to use this LP with the number
of rounds $r$ being $\max_{a \in X} 2|V_a|$, but solving this LP would
require time $n^{O(k \log n)}$, which is undesirable. Hence, we write an
LP that uses only some of the variables from $\eqref{eqn:sa_sc}$.  Let
$\ess_a$ denote the power set of $V_a$. Let $\ess_{ab}$ be the power set
of $V_a \cup V_b$ and let $\ess := \cup_{a ,b \in X} \ess_{ab}$. For
every set $S \in \ess$, and every subset $T \subseteq S$, we retain the
variable $x(S,T)$ in the LP, and drop all the others. There are at most
$\poly(n)$ nodes in $X$, and hence $\poly(n)$ sets $\ess_{ab}$, each of
these has at most $2^{2k\lambda} = n^{O(k)}$ many sets. This results in
an LP with $n^{O(k)}$ variables and a similar number of constraints.

Finally, as mentioned above, to take care of the non-linear objective
function in~(\ref{eqn:sa_sc}), we guess the optimal value $\alpha^\star
> 0$ of the denominator, and add the constraint 
\[
\sum_{(u,v) \in E_D} \dem_{uv} y_{uv} \geq \alpha^\star
\]
as an additional constraint to the
LP, thereby just minimizing \ifstoc \\ \fi $\sum_{(u,v) \in E_G} \capa_{uv} y_{uv}$. For
the rest of the discussion, let $(x,y)$ be an optimal solution to the
resulting LP.

\subsection{The Rounding Algorithm}

The rounding algorithm is a very natural top-down propagation rounding
procedure.  We start with the root $r \in X$; note that $V_r = U_r$ in
this case.  Since $\sum_{S \subseteq V_r} x(V_r,T) = 1$ by the
constraints~(\ref{eqn:sa_sum_to_1}) of the LP, the $x$ variables define
a probability distribution over subsets of $V_r$. We sample a subset
$A_r$ from this distribution.

In general, for any node $a \in X$ with parent $b$, suppose we have
already sampled a subset for each of its ancestor nodes $b, \cdots, r$,
and the union of these sampled sets is $A_b \sse V_b$. Now, let $\bee_a =
\{ A' \sse V_a \mid A' \cap V_b = A_b\}$; i.e., the family of subsets of
$V_a$ whose intersection with $V_b$ is precisely $A_b$. By
\lref[Lemma]{lem:cons}, we have 
\[ x(V_b,A_b) = \sum_{A' \in \bee_a} x(V_a,A'). \] Thus the values
$x(V_a,A')/x(V_b,A_b)$ define a probability distribution over $\bee_a$.
We now sample a set $A_a$ from this distribution. Note that this
rounding only uses sets we retained in our pared-down LP, so we can
indeed implement this rounding. Moreover, this set $A_a \supseteq A_b$.
Finally, we take the union of all the sets 
\[ A := \cup_{a \in X} A_a, \] and output the cut $(A,\Abar)$. The
following lemma is immediate:
\begin{lemma}
  \label{lem:true-prob}
  For any $a \in X$ and any $S\in \ess_a$, we get $\Pr[( A \cap S) = T]
  = x(S,T)$ for all $T \sse S$.
\end{lemma}
\begin{proof}
  First, we claim that $\Pr[ A_a = T ] = x(V_a, T)$ for all $a \in X$.
  This is a simple induction on the depth of $a$: the base case is directly from the
  algorithm. For $a \in X$ with parent node $b$,
 \ifstoc
  \begin{align*}
   \Pr[ A_a = T] &= \Pr[ A_b = T \cap V_b] \cdot \Pr[ A_a = T \mid A_b = T \cap V_b] \\
  &= x(V_b, T \cap V_b) \cdot \frac{x(V_a, T)}{x(V_b, T \cap V_b)} \\
  &= x(V_a, T),
 \end{align*}
 \else
  \[
  \Pr[ A_a = T] = \Pr[ A_b = T \cap V_b] \cdot \Pr[ A_a = T \mid A_b
  = T \cap V_b] = x(V_b, T \cap V_b) \cdot \frac{x(V_a, T)}{x(V_b, T
    \cap V_b)} = x(V_a, T),
 \]
 \fi
    as claimed.  Now we prove the statement
  of the lemma: Since $S \sse V_a$, we know that $\Pr[ A \cap S = T ] =
  \Pr[ A_a \cap S = T ]$, because none of the future steps can add any
  other vertices from $V_a$ to $A$. Moreover,
  \ifstoc
  \begin{align*}
    \Pr[ A_a \cap S = T ] &= \sum_{T' \sse V_a \setminus S} \Pr[ A_a = T \cup T' ] \\
  &= \sum_{T' \sse V_a \setminus S} x(V_a, T \cup T'),
  \end{align*}
  \else
  \[
  \Pr[ A_a \cap S = T ] = \sum_{T' \sse V_a \setminus S} \Pr[ A_a = T
  \cup T' ] = \sum_{T' \sse V_a \setminus S} x(V_a, T \cup T'),
  \] 
  \fi
  the last equality using the claim above. Defining $S' := V_a \setminus S$,
  this equals $\sum_{T' \sse S'} x(S \cup S', T \cup T')$, which by
  \lref[Lemma]{lem:cons} equals $x(S,T)$ as desired. 
\end{proof}

\begin{lemma}
  \label{lem:edge-cut}
  The probability of an edge $(u,v) \in E_G$ being cut by $(A, \Abar)$
  equals $y_{uv}$.
\end{lemma}
\begin{proof}
  By the properties of tree-decompositions, each edge $(u,v) \in E_G$
  lies within $U_a$ for some $a \in X$, and $\{u,v\} \sse \ess_a$. The
  probability of the edge being cut is
  \ifstoc
  \begin{align*}
  \Pr[ A \cap \{u,v\} = \{u\}] + \Pr[ &A \cap \{u,v\} = \{v\}] \\
  &= x(\{u,v\}, \{u\}) + x(\{u,v\}, \{v\}) \\
  &= y_{uv}.
  \end{align*}
  \else
  \[
  \Pr[ A \cap \{u,v\} = \{u\}] + \Pr[ A \cap \{u,v\} = \{v\}] = x(\{
  u,v\}, \{u\}) + x(\{u,v\}, \{v\}) = y_{uv}.
  \]
  \fi
  The first equality above follows from \lref[Lemma]{lem:true-prob},
  and the second from the definition of $y_{uv}$ in~(\ref{eqn:sa_edges}).
\end{proof}
Thus the expected number of edges in the cut $(A, \Abar)$ equals the
numerator of the objective function.

\begin{lemma}
  \label{lem:dem-cut}
  The probability of a demand pair $(s,t) \in E_D$ being cut by $(A,
  \Abar)$ is at least $y_{st}/2$.
\end{lemma}
%New proof below commented out the old proof for now.
\begin{proof}
  Let $a, b \in X$ denote the (least depth) nodes in $\T$ such that $s
  \in U_a$ and $t \in U_b$ respectively; for simplicity, assume that the
  least common ancestor of $a$ and $b$ is $r$. (An identical argument
  works when the least common ancestor is not the root.) We can assume
  that $r \notin \{a,b\}$, or else we can use \lref[Lemma]{lem:true-prob} to
  claim that the probability $s,t$ are separated is exactly $y_{st}$.

  Consider the set $V_a \cup V_b$, and consider the set-valued random
  variable $W$ (taking on values from the power set of $V_a \cup V_b$)
  defined by $\Pr[ W = T ] := x(V_a \cup V_b, T)$. Denote the
  distribution by $\dee_{ab}$, and note that this is just the
  distribution specified by the \SA LP restricted to $V_a \cup V_b$. Let
  $X_s$ and $X_t$ denote the indicator random variables of the events $\{s \in W\}$
  and $\{t \in W\}$ respectively; these variables are dependent in
  general. For a set $T \subseteq V_r$, let $X_{s|T}$ and $X_{t|T}$ be
  indicators for the corresponding events conditioned on $W \cap V_r =
  T$. Then by definition, 
  \begin{gather}
    y_{st} = \Pr_{\dee_{ab}}[X_s \neq X_t]
    = \E_T \Pr_{\dee_{ab}}[X_{s|T} \neq X_{t|T}]
  \end{gather}
   where the expectation is taken over outcomes of $T = W \cap V_r$.

  Let $\dee$ denote the distribution on cuts defined by the
  algorithm. Let $Y_s$ and $Y_t$ denote events that $\{s \in A\}$ and
  $\{t \in A\}$ respectively, and let $Y_{s|T}$ and $Y_{t|T}$ denote
  these events conditioned on $A\cap V_r=T$. Thus the probability that
  $s$ and $t$ are separated by the algorithm is
  \begin{gather}
    \ALG(s,t) = \Pr_{\dee}[Y_s \neq Y_t]
    = \E_T \Pr_{\dee}[Y_{s|T} \neq Y_{t|T}]
  \end{gather}
  where the expectation is taken over the distribution of $T = A \cap
  V_r$; by \lref[Lemma]{lem:true-prob} this distribution is the same as
  that for $W \cap V_r$.

  It thus suffices to prove that for any $T$,
  \begin{align}
    \Pr_{\dee_{ab}}[X_{s|T} \neq X_{t|T}] \leq 2 \Pr_{\dee}[Y_{s|T} \neq
    Y_{t|T}]. 
    \label{eq:depvsind}\end{align}
  Now observe that $Y_{s|T}$ is distributed identically to $X_{s|T}$
  (with both being 1 with probability $\frac{x(V_r \cup \{s\},T \cup
    A)}{x(V_r,T)}$), and similarly for $Y_{t|T}$ and $X_{t|T}$. However,
  since $s$ and $t$ lie in different subtrees, $Y_{s|T}$ and $Y_{t|T}$
  are independent, whereas $X_{s|T}$ and $X_{t|T}$ are dependent in
  general. 

  We can assume that at least one of $\E_{\dee_{ab}}[X_{s|T}],
  \E_{\dee_{ab}}[X_{t|T}]$ is at most $1/2$; if not, we can do the
  following analysis with the complementary events
  $\E_{\dee_{ab}}[X_{\overline{s}|T}],
  \E_{\dee_{ab}}[X_{\overline{t}|T}]$, since~(\ref{eq:depvsind})
  depends only on random variables being unequal. Moreover, suppose
  \[ \E_{\dee_{ab}}[X_{t|T}] \leq \E_{\dee_{ab}}[X_{s|T}] \] (else we
  can interchange $s,t$ in the following argument). Define the
  distribution $\dee'$ where we draw $X_{s|T}, X_{t|T}$ from
  $\dee_{ab}$, set $Y_{s|T}$ equal to $X_{s|T}$, and draw $Y_{t|T}$
  independently from $\dee$. By construction, the distributions of
  $X_{s|T},X_{t|T}$ in $\dee_{st}$ and $\dee'$ are identical, as are the
  distributions of $Y_{s|T}, Y_{t|T}$ in $\dee$ and $\dee'$. We claim
  that
  \begin{align}
    \E_{\dee'}[X_{t|T} \neq Y_{t|T}] \leq \E_{\dee'}[X_{s|T} \neq
    Y_{t|T}].\label{eq:step1} 
  \end{align}
  Indeed, if $\E_{\dee'}[X_{s|T}] = a$ and $\E_{\dee'}[X_{t|T}] = b$, then
  $\E_{\dee'}[Y_{t|T}] = b$ as well, with $b \leq a$ and $b \leq
  1/2$. Thus,~(\ref{eq:step1}) claims that $2b(1-b) \leq a(1-b) +
  b(1-a)$ (recall here that $Y_{t|T}$ is chosen independently of the
  other variables).  This holds if $b(1 - 2b) \leq a(1 - 2b)$, which follows from
  our assumptions on $a,b$ above. Finally,
  \begin{align}
    \Pr_{\dee'}[X_{s|T} \neq X_{t|T}] \leq \Pr_{\dee'}[X_{s|T} \neq
    Y_{t|T}] + \Pr_{\dee'}[X_{t|T} \neq Y_{t|T}]. \label{eq:step2} 
  \end{align}
  Combining (\ref{eq:step1}) and (\ref{eq:step2}) and observing that
  $X_{s|T}=Y_{s|T}$ in our construction, the claim follows. 
\end{proof}

By \lref[Lemmas]{lem:edge-cut} and~\ref{lem:dem-cut}, a random
cut $(A, \Abar)$ chosen by our algorithm cuts an expected capacity of
exactly $\sum_{uv \in E_G} \capa_{uv} y_{uv}$, whereas the expected demand
cut is at least $\frac12 \sum_{st \in E_D} \dem_{st} y_{st}$. This shows
the \emph{existence} of a cut in the distribution whose sparsity is within a
factor of two of the LP value. 
Such a cut can be found using
the method of conditional expectations; we defer the details to the 
\stocoption{final version}{next section}.
Moreover, the analysis of the
integrality gap is tight: 
\stocoption{the full version}{\lref[Section]{sec:SA-gaps}} shows that for any constant $\gamma > 0$,
the \SA LP for Sparsest Cut has an integrality gap of at least $2 -
\eps(\gamma)$, even after $n^\gamma$ rounds.

\stocoption{}{
\ifstoc
\section{Derandomization}
\else
\subsection{Derandomization}
\fi
\label{sec:derand}

\ifstoc Recall that the rounding algorithm guarantees that the expected
capacity of edges cut is the same as what the LP pays (fractionally),
and the expected demand cut is at least half of what the LP cuts
(fractionally). This shows the existence of a good cut, but does not
immediately give us such a cut.
\fi
In this section, we use the method of conditional expectations to
derandomize our rounding algorithm, which allows us to efficiently find
a cut $(A, \Abar)$ with sparsity at most twice the LP value. We will
think of the set $A$ as being a $\{0,1\}$-assignment/labeling for the
nodes in $V$, where $i \in A \iff A(i) = 1$.

In the above randomized process, let $Y_{ij}$ be the indicator random
variable for whether the pair $(i,j)$ is separated. We showed that for
$(i,j) \in E_G$, $\E[Y_{ij}] = y_{ij}$, and for all other $(i,j) \in
\binom{V}{2}$, $\E[Y_{ij}] \geq y_{ij}/2$. Now if we let $Z = \sum_e \capa_e
Y_e$ be the r.v.\ denoting the edge capacity cut by the process and $Z'
= \sum_{st} \dem_{st} Y_{st}$ be the r.v.\ denoting the demand
separated, then the analysis of the previous section shows that 
\[ \frac{\E[Z]}{\E[Z']} \leq 2 \cdot \frac{\sum_e \capa_e
  y_e}{\alpha}. \] (Recall that $\alpha$ was the ``guessed'' value of
the total demand separated by the actual sparsest cut.) Equivalently,
defining $\LPOpt := \sum_e \capa_e y_e$, and
\[
W := \frac{Z}{\LPOpt} - \frac{2\, Z'}{\alpha},
\]
we know that $\E[W] \leq 0$.

The algorithm is the natural one: for the root $r$, enumerate over all
$2^k$ assignments for the bag $V_r$, and choose the assignment $A_r$
minimizing $\E[W \mid A_r]$.  Since $\E[W] \leq 0$, it must be the case
that $\E[W \mid A_r] \leq 0$ by averaging. Similarly, given the choices
for nodes $X' \sse X$ such that $\T[X']$ induces a connected tree and 
$\E[W \mid \{A_x\}_{x \in X'}] \leq 0$, choose any $a \in X$ whose
parent $b \in X'$, and choose an assignment $A_a$ for the nodes in $V_a
\setminus V_b$ so that the new $\E[W \mid \{A_x\}_{x \in X' \cup
  \{a\}}] \leq 0$. The final assignment $A$ will satisfy $\E[W \mid
\{A_a\}_{a \in X}] \leq 0$, which would give us a cut with sparsity at
most $2\LPOpt/\alpha$, as desired.

It remains to show that we can compute $\E[W \mid \{A_x\}_{x \in X'}]$
for any subset $X' \sse X$ containing the root $r$, such that $\T[X']$
is connected. Let $V' = \cup_{x \in X'} U_x$ be the set of nodes already
labeled. For any vertex $v \in V$, let $b(v) \in X$ be the highest node
in $\T$ such that $v \in U_{b(v)}$. If $v$ is yet unlabeled, then $b(v)
\notin X'$, and hence let $\ell(v)$ be the lowest ancestor of $b(v)$ in
$X'$. In other words, we have chosen an assignment $A_{\ell(v)}$ for the
bag $V_{\ell(v)}$. By the properties of our algorithm, we know that
\begin{gather}
 {\small \Pr[ v \in A \mid A_{\ell(v)} ] = \frac{x( V_{\ell(v)} \cup \{v\},
    A_{\ell(v)} \cup \{v\})}{x( V_{\ell(v)}, A_{\ell(v)})}.} \label{eq:6}
\end{gather}
Moreover, if $u, v$ are both unlabeled such that their highest bags
$b(u), b(v)$ share a root-leaf path in $\T$, then
\ifstoc
{\small \begin{align*}
  \Pr[ u, v \text{ separated} \mid A_{\ell(v)} ] &= \notag 
  \frac{x( V_{\ell(v)} \cup \{u,v\}, A_{\ell(v)} \cup \{u\})}{x( V_{\ell(v)}, A_{\ell(v)})} \notag \\
  & \qquad + \frac{x( V_{\ell(v)} \cup \{u,v\}, A_{\ell(v)} \cup \{v\})}{x( V_{\ell(v)}, A_{\ell(v)})}, \label{eq:7}
\end{align*}}
\else
\begin{gather}
  \Pr[ u, v \text{ separated} \mid A_{\ell(v)} ] = \frac{x( V_{\ell(v)}
    \cup \{u,v\}, A_{\ell(v)} \cup \{u\}) + x( V_{\ell(v)} \cup \{u,v\},
    A_{\ell(v)} \cup \{v\})}{x( V_{\ell(v)}, A_{\ell(v)})}, \label{eq:7}
\end{gather}
\fi
where $\ell(v) = \ell(u)$ is the lowest ancestor of $b(u), b(v)$ that
has been labeled. If $u,v$ are yet unlabeled, but we have chosen an
assignment for $a = \lca(b(u), b(v))$, then $u, v$ will be labeled
independently using~(\ref{eq:6}). Finally, if $u,v$ are unlabeled, and
we have not yet chosen an assignment for $a = \lca(b(u), b(v))$, then
the probability of $u,v$ being cut is precisely
\ifstoc
\begin{align*}
\sum_{U \sse V_a \setminus V_{\ell(v)}} & \left( \frac{ x(V_a, A_{\ell(v)} \cup U) }{x(V_{\ell(v)}, A_{\ell(v)})} \right. \\
& \quad \cdot \left. \Pr[ (u,v) \text{ separated } \mid V_a \text{ labeled } A_{\ell(v)} \cup U ] \vphantom{\frac{ x(V_a, A_{\ell(v)} \cup U) }{x(V_{\ell(v)}, A_{\ell(v)})}} \right),
\end{align*}
\else
\[
\sum_{U \sse V_a \setminus V_{\ell(v)}} \frac{ x(V_a, A_{\ell(v)}
  \cup U) }{x(V_{\ell(v)}, A_{\ell(v)})} \cdot \Pr[ (u,v) \text{
  separated } \mid V_a \text{ labeled } A_{\ell(v)} \cup U ],
\]
\fi
where the probability can be computed using~(\ref{eq:6}), since $u, v$ will be
labeled independently after conditioning on a labeling for $V_a$.  There
are at most $n^{O(k)}$ terms in the sum, and hence we can compute this
in the claimed time bound. Now, we can compute $\E[ W \mid \{A_x\}_{x
  \in X'} ]$ using the above expressions in time $n^{O(k)}$, which
completes the proof.

\subsubsection{Embedding into $\ell_1$}
\label{sec:embedding}

Our algorithm and analysis also implies a $2$-approximation to the minimum distortion $\ell_1$ embedding of a treewidth $k$ graph in time $n^{O(k)}$. We will describe an algorithm that, given $D$, either finds an embedding with distortion $2D$ or certifies that any $\ell_1$ embedding of $G$ requires distortion more than $D$. It is easy to use such a subroutine to get a $2+o(1)$-approximation to the minimum distortion $\ell_1$ embedding problem.

Towards this end, we write a relaxation for the distortion $D$ embedding problem as follows. Given $G$ with treewidth $k$, we start with the $r$-round Sherali-Adams polytope $\mathbf{SA}_r(n)$ with $r=O(k\log n)$. We add the additional set of constraints  $C \cdot d(u,v) \leq y_{uv} \leq D\cdot C \cdot d(u,v)$, for every pair of vertices $u,v \in V$. The cut characterization of $\ell_1$ implies that this linear program is feasible whenever there is a distortion $D$ embedding.  Given a solution to the linear program, we round it using the rounding algorithm of the last section. It is immediate from our analysis that a random cut sampled by the algorithm satisfies $Pr[(u,v) \mbox{ separated}] \in [y_{uv}/2,y_{uv}]$.

Moreover, since the analysis of the rounding algorithm only uses
$n^{O(k)}$ equality constraints on the expectations of random variables,
we can use the approach of Karger and Koller~\cite{KK94} to get an
explicit sample space $\Omega$ of size $|\Omega| = n^{O(k)}$ that
satisfies all these constraints. Indeed, each of the points $\omega \in
\Omega$ of this sample space gives us a $\{0,1\}$-embedding of the
vertices of the graph. We can concatenate all these embeddings and scale
down suitably in time $|\Omega|\cdot \poly(n)$ to get an
$\ell_1$-embedding $f: V \to \R^{|\Omega|}$ with the properties that
(a)~$\|f(u) - f(v)\|_1 = y_{uv}$ for all $(u,v) \in E_G$, and
(b)~$\|f(u) - f(v)\|_1 \geq y_{uv}/2$ for $(u,v) \in {V\choose 2}$. 
Scaling $f$ by a factor of $C$  gives an embedding with distortion $2D$.

%%% Local Variables: 
%%% mode: latex
%%% TeX-master: "sp-cut"
%%% End: 
}

%%% Local Variables: 
%%% mode: latex
%%% TeX-master: "sp-cut"
%%% End: 

\section{The Hardness Result} 
\label{sec:hardness}

In this section, we prove the \textsc{Apx}-hardness claimed in
\lref[Theorem]{thm:main1}. In particular, we show the following reduction
from the \mc problem to the \nusc problem.
\begin{theorem}
  \label{thm:maxcut-hard}
  For any $\eps > 0$, a $\rho$-approximation algorithm for \nusc on
  series-parallel graphs (with arbitrary demand graphs) that runs in
  time $T(n)$ implies a \ifstoc $(\frac{1}{\rho} - \eps)$-approxi-mation \else $(\frac{1}{\rho} - \eps)$-approximation \fi to \mc on
  general graphs running in time $T(n^{O(1/\eps)})$.
\end{theorem}
The current best hardness-of-approximation results for \mc are: (a)~the
$(\frac{16}{17} + \eps)$-factor hardness (assuming $\text{P} \neq \text{NP}$) due to
H{\aa}stad~\cite{Hastad01} (using the gadgets from Trevisan et
al.~\cite{TSSW}) and (b)~the $(\alpha_{GW} - \eps)$-factor hardness
(assuming the Unique Games Conjecture) due to Khot et
al.~\cite{KKMO,MOO}, where $\alpha_{GW} = 0.87856\ldots$ is the constant
obtained in the hyperplane rounding for the \mc SDP. Combined with
\lref[Theorem]{thm:maxcut-hard}, these imply hardness results of
$(\frac{17}{16} - \eps)$ and $(1.138 - \eps)$ respectively for \nusc
and prove \lref[Theorem]{thm:main1}.

The proof of \lref[Theorem]{thm:maxcut-hard} proceeds by taking the hard
\mc instances and using them to construct the demand graphs in a \spcut
instance, where the supply graph is the familiar fractal obtained from
the graph $K_{2,n}$.\footnote{The fractal for $K_{2,2}$ has been used
  for lower bounds on the distortion incurred by tree
  embeddings~\cite{GNRS99}, Euclidean embeddings~\cite{NR02}, and
  low-dimensional embeddings in $\ell_1$~\cite{BC03, LN03, Regev12}.
  Moreover, the fractal for $K_{2,n}$ shows the integrality gap for the
  natural metric relaxation for \spcut~\cite{LR10, CSW10}.}  The base
case of this recursive construction is in \lref[Section]{sec:basic-bb},
and the full construction is in \lref[Section]{sec:recursive}. The
analysis of the latter is based on a generic powering lemma, which will
be useful for showing tight Unique Games
hardness for bounded treewidth graphs in \lref[Section]{sec:ug-hardness} and
the \SA integrality gap in
\stocoption{the full version}{\lref[Section]{sec:SA-gaps}}.
%  with the
% proof of correctness in \lref[Section]{sec:proof-hardness}.

\subsection{The Basic Building Block}
\label{sec:basic-bb}

Given a connected (unweighted) \mc instance $H = ([n], E_H)$, let $m
= |E_H|$, and let $\maxcut(H) := \max_{A \subseteq [n]}
|\partial_H(A)|$.  Let the supply graph be $G_1' = (V_1, E_1)$, with
vertices $V_1 = \{s, t\} \cup [n]$ and edges $E_1 = \cup_{i \in [n]} \{
\{s, i\}, \{t, i\}\}$.  Define the capacities $\capa_{s,i} = \capa_{t,i}
= \deg_H(i)/2m$. Define the demands thus: $\dem_{s,t} = 1$, and for $i,j
\in [n]$, let $\dem_{i,j} = \mathbf{1}_{\{i,j\} \in E_H}/m$ (i.e., $i,j$
have $1/m$ demand between them if $\{i,j\}$ is an edge in $H$, and zero
otherwise). Let this setting of demands be denoted $D_1'$. (The hardness
results in Chuzhoy and Khanna~\cite{CK06} and Chlamt\'a\v{c} et
al.~\cite{CKR10} used the same graph $G_1$, but with a different choice
of capacities and demands.)

\begin{claim}
  \label{clm:separating-best}
  The sparsest cuts in $G_1'$ are $s$-$t$-separating, and have sparsity
  $m/(m + \maxcut(H))$.
\end{claim}

\begin{proof}
  For $A \subseteq [n]$, the cut $(A + s, \Abar + t)$ has sparsity
  \[ \frac{ \sum_{i \in [n]} \deg_H(i)/2m }{ |\partial_H(A)|/m + 1 } = \frac{
    m}{ |\partial_H(A)| + m } < 1, \] since $\frac12 \sum_i \deg_H(i) =
  m$. The cut $(A, \Abar + s + t)$ has sparsity \[ \frac{ 2 \sum_{i \in
      A} \deg_H(i)/2m }{ |\partial_H(A)|/m } \geq \frac{ 2 \sum_{i \in A}
    \deg_H(i)/2 }{ \sum_{i \in A} \deg_H(i) } \geq 1, \] which is strictly
  worse than any $s$-$t$-separating cut. Hence the sparsest cut is the
  cut $(A + s, \Abar + t)$ that maximizes $|\partial_H(A)|$.
\end{proof}
Given a $cm$-vs-$sm$ hardness result for \mc, this gives us a $(1 +
c)$-vs-$(1+s)$ hardness for Sparsest Cut.  However, we can do better
using a recursive ``fractal'' construction, as we show next. Before we proceed further, we remark that if we remove the $s$-$t$ demand from the instance $G'_1$, we obtain an instance $G_1$ with the following properties.
\begin{lemma}
\label{lem:buildingblock}
The instance $G_1$ constructed by removing $\dem_{s,t}$ from $G_1'$ satisfies:
\begin{OneLiners}
\item If $H$ has a cut of size $cm$, then there is an $s$-$t$ separating cut of capacity $1$ that separates $c$ demand.
\item Any $s$-$t$ separating cut has capacity at least $1$.
\item If the maximum cut in $H$ has size $sm$, then every $s$-$t$ separating cut has sparsity at least $s^{-1}$.
\item Any cut that does not separate $s$ and $t$ has sparsity at least $1$.
\end{OneLiners}
\end{lemma}
While $G_1$ by itself is not a hard instance of Sparsest Cut, the above
properties will make it a useful building block in the powering
operation below.  

\begin{figure}[ht]
  \begin{centering}
    \includegraphics[scale=0.5]{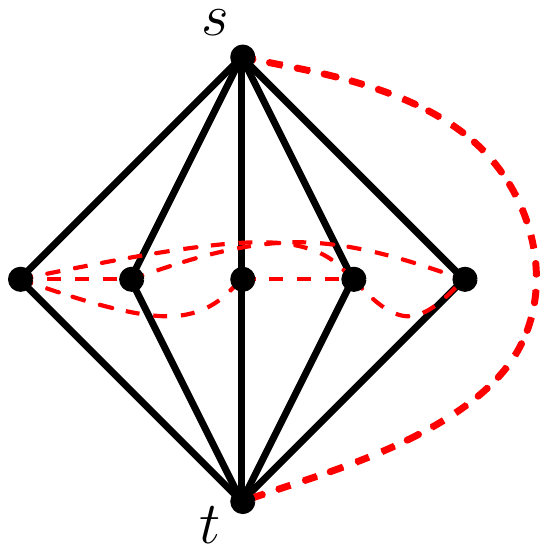}
    \label{fig:base-case}
    \caption{Base case of the construction $G'_1$ for $n=5$.}
  \end{centering}
\end{figure}

\subsection{An Instance Powering Operation}
\label{sec:recursive}

In this section, we describe a powering operation on Sparsest Cut
instances that we use to boost the hardness result.  This is
the natural fractal construction. We start with an instance $G_1=(V_1 =
\{s,t\} \cup [n], \capa_e, \dem_e)$ of the sparsest cut problem. In
other words, we have a Sparsest Cut instance with two designated vertices
$s$ and $t$. (For concreteness, think of the $G_1$ from the previous
section, but any graph $G_1$ would do.)

For $\ell \geq 2$, consider the graph $G_{\ell}$ obtained by taking
$G_1$ and replacing each capacity edge $e=(u,v)$ in $G_1$ with a copy of $G_{\ell-1}$ in
the natural way. In other words, for every $e=(u,v)$, we create a copy $G_{\ell-1}^e$  of $G_{\ell-1}$, and identify its vertex $s$ with $u$ and its $t$ with $v$. Moreover, $G_{\ell-1}^{e}$ is scaled down by $\capa_e$. Thus if edge $f \in E_{\ell-1}$ has capacity $\capa_f$
in $G_{\ell-1}$, then the corresponding edge in $G_{\ell-1}^{e}$ has capacity $\capa_e \cdot \capa_f$;
the demands in $G_{\ell-1}^{e}$ are also scaled by the same factor. 
In addition to the scaled demands from copies of
$G_{\ell-1}$, $G_\ell$ contains new {\em level-$\ell$} demands $\dem_{i,j}$  from the base graph $G_1$. Note that this instance contains vertices of $V_1$ in its vertex set and will have $s$ and $t$ as its designated vertices.

The following properties are immediate.
\begin{observation} If $G_1$ has $n$ vertices and $m$ capacity edges, then $G_{\ell}$ has $m^{\ell-1}n$ vertices and $m^{\ell}$ capacity edges. Moreover, if the supply graph in $G_1$ has treewidth $k$, then the supply graph of $G_{\ell}$ also has treewidth $k$.
\end{observation}

We next argue ``completeness'' and ``soundness'' properties of this operation. We will distinguish between cuts that separate $s$ and $t$, and those that do not. We call the
former cuts {\em admissible} and the latter {\em inadmissible}.

\begin{lemma}
  If $G_1$ has an admissible cut $(A,\overline{A})$ that cuts $\capa(A,\overline{A})$ capacity and $\dem(A,\overline{A})$ demand, then there exists an admissible cut in $G_{\ell}$ of capacity $(\capa(A,\overline{A}))^\ell$ that cuts $\dem(A,\overline{A})\cdot (\sum_{i=0}^{\ell-1} \capa(A,\overline{A})^i)$ demand.
  \label{lem:powercompleteness}
\end{lemma}
\begin{proof}
  The proof is by induction on $\ell$.  The base case $\ell=1$ is
  an assumption of the lemma. Assume the claim holds for $G_{\ell-1}$. Let
 $\cut{A_{\ell-1}}$ denote the admissible cut satisfying the
  induction hypothesis and let $s \in A_{\ell-1}$. Recall that $G_\ell$
  is created by replacing the edges of $G_1$ by copies of $G_{\ell-1}$.
  Define the cut $A_\ell$ in the natural way: Start with
  $A_{\ell} = A$. Then for each $e=(u,v) \in G_1$ such that $u,v \in A$,
  we place all of $G_{\ell-1}^e$ in $A_{\ell}$; similarly if $u,v \in
  \Abar$ then place all of $G_{\ell-1}^e$ in $\Abar_{\ell}$. For $(u,v) \in G_1$ such that $u\in A, v \in
  \overline{A}$, we cut $G_{\ell-1}^e$ according to $\cut{A_{\ell-1}}$:
  i.e., the copy of a vertex $x \in A_{\ell-1}$ is placed in $A_\ell$.
  Similarly, if $u \in \overline{A}, v \in A$, we put the
  copy of $x$ in $A_{\ell}$ if $x \in \overline{A_{\ell-1}}$. This defines the cut
  $\cut{A_\ell}$.

  The capacity of the cut can be computed as follows: For each edge of
  $G_1$ cut by $A$, the corresponding copy of $G_{\ell-1}$ contributes
  $\capa_e \cdot \capa(A_{\ell-1}) = \capa_e \cdot
  (\capa(A,\overline{A})^{\ell-1})$ to the cut, where we used the
  inductive hypothesis for $A_{\ell-1}$. For edges not cut by $\cut{A}$,
  the corresponding $G_{\ell-1}^e$ is uncut and contributes~0. Thus
  \begin{gather*}
    \ts \capa(A_\ell,\Abar_\ell) = \sum_{e \in (A, \Abar)} \capa_e \cdot
    (\capa(A,\overline{A})^{\ell-1}) = \capa(A,\overline{A})^{\ell}.
  \end{gather*}
  Similarly, the demand from copies of $G_{\ell-1}$ cut by $A_\ell$ is
  exactly
  \begin{align*}
    & \ts  \sum_{e \in (A, \Abar)} \capa_e \cdot
    \dem(A_{\ell-1},\Abar_{\ell-1}) \\  & \ts = \capa\cut{A} \cdot
    \dem(A_{\ell-1},\Abar_{\ell-1}) \\
    & \ts = \capa\cut{A} \cdot \dem(A,\overline{A})\cdot (\sum_{i=0}^{\ell-2}
    \capa(A,\overline{A})^i) \\ & = \ts \dem(A,\overline{A})\cdot
    (\sum_{i=1}^{\ell-1} \capa(A,\overline{A})^i).
  \end{align*}
  (The second equality is from the induction hypothesis.) Additionally,
  $A_{\ell}$ cuts exactly $\dem\cut{A}$ units of the level-$\ell$
  demands. The claim follows by the summing the two.
\end{proof}

Note that if $G_1$ has an admissible cut $(A,\overline{A})$ of capacity
$1$ that cuts $\dem(A,\overline{A})$ units of demand, then the above
lemma gives us a cut of capacity $1$ that cuts $\ell\,
\dem(A,\overline{A})$ units of demand.

Now, for soundness analysis, we argue that if $G_1$ has no ``good''
cuts, then neither does $G_\ell$. It will be convenient to separately
argue about the admissible and inadmissible cuts.

We will need the notion of ``connected'' cuts. Given a graph $G =
(V,E)$, call a cut $(X, V \setminus X)$ \emph{connected} if the
resulting components $G[X]$ and $G[V \setminus X]$ are \emph{both}
connected graphs. Observe that for a connected admissible cut $(A + s,
\Abar + t)$ in $G_\ell$, along any $s$-$t$ shortest path $P$, the
vertices in $P \cap (A +s)$ forms some prefix of $P$---this path is cut
exactly once.

\begin{lemma}[\cite{OS81}, Lemma~2.1(ii)]
  \label{lem:os}
  For any connected Sparsest Cut instance $(G,D)$, there exists a sparsest cut
  that is connected.
\end{lemma}

\stocoption{}{
\begin{proof}
Let $S$ be the sparsest cut in $G$ such that $\sum_{e \in \partial_G(S)} \capa_e$ is as small as possible. We claim that $S$ must be connected. Suppose not, and let $S_1, \ldots, S_k$ be the partition of $S$ into connected components with $k \geq 3$.  Let $\capa_e$ be the capacity on edge $e \in E_G$ and let $\dem_e$ be the demand on edge $e \in E_D$.  Let $\Phi_{G,D}$ be the value of the Sparsest Cut instance.  Then
  \[
  \Phi_{G,D} = \frac{\sum_{e \in \partial_G(S)} \capa_e}{\sum_{e \in \partial_\jay(S)} \dem_e } \geq \frac{\sum_{i=1}^{k} \sum_{e \in \partial_G(S_i)} \capa_e}{\sum_{i=1}^k \sum_{e \in \partial_\jay(S_i)} \dem_e }.
  \]
Since $S$ is a sparsest cut, each $S_i$ has sparsity at least $\Phi_{G,D}$. It follows that in fact all $S_i$'s have the same sparsity. Thus for each $i$, 
\[
  \Phi_{G,D} = \frac{\sum_{e \in \partial_G(S_i)} \capa_e}{\sum_{e \in \partial_\jay(S_i)} \dem_e}.
 \]
 Since $G$ is connected, each of the $\sum_{e \in \partial_G(S_i)}
 \capa_e$ quantities are positive. Moreover, since $k \geq 3$, there
 must be an $i$ such that $\sum_{e \in \partial_G(S_i)} \capa_e$ is
 strictly smaller than $\sum_{e \in \partial_G(S)} \capa_e$. This,
 however, contradicts the definition of $S$, and the claim follows.
\end{proof}
}

We now proceed to main technical result of this section.

\begin{lemma}
Suppose that for some constant $\gamma$, $G_1$ satisfies:
  \begin{OneLiners}
  \item Any admissible cut $(A,\Abar)$ has capacity $\capa(A,\Abar)$ at
    least $1$.
  \item Any admissible cut $(A,\Abar)$ cuts at most $\gamma \cdot
    \capa(A,\Abar)$ demand.
  \item Any inadmissible cut $(A,\Abar)$ cuts at most $\capa(A,\Abar)$ demand.
  \end{OneLiners}
Then $(G_{\ell},D_{\ell})$  satisfies:
  \begin{OneLiners}
  \item Any admissible cut $(A,\Abar)$ has capacity $\capa(A,\Abar)$ at
    least $1$.
  \item Any admissible cut $(A,\Abar)$ cuts at most $\ell \gamma \cdot
    \capa(A,\Abar)$ demand.
  \item Any inadmissible cut $(A,\Abar)$ cuts at most $((\ell-1)\gamma +
    1) \cdot \capa(A,\Abar)$ demand.
  \end{OneLiners}
  \label{lem:powersoundness}
\end{lemma}
\begin{proof}
  The proof is by induction on $\ell$. The base case $\ell=1$ is the
  assumption of the lemma. Suppose that the claim holds for
  $G_{\ell-1}$.
  
  First, let $(A_{\ell},\Abar_{\ell})$ be an admissible cut. Let
  $(A_1,\Abar_1)$ denote the projection of this cut onto $\{s,t\} \cup
  [n]$, i.e., $A_1 = A_\ell \cap ([n] \cup \{s,t\})$. For each edge $e
  \in (A_1,\Abar_1)$, the cut $A_{\ell}$ induces an admissible cut on
  $G_{\ell-1}^e$. This contributes at least unit capacity to the
  corresponding level-$(\ell-1)$ cut (by the induction hypothesis), and
  thus $\capa_e \cdot 1$ to the cut $(A_{\ell},\Abar_{\ell})$ because of
  the scaling-down in the construction of $G_\ell$. Summing 
  over all edges $e \in (A_1, \Abar_1)$, we conclude that
  $\capa\cut{A_{\ell}}$ is at least the capacity of $\cut{A_1}$ in
  $G_1$. Using the fact that all admissible cuts in $G_1$ have capacity
  at least $1$, the first part of the claim follows.

  Next, we estimate the demand cut by $(A_{\ell},\Abar_{\ell})$.  The
  total level $\ell$ demand cut is at most $\gamma$ times the capacity of $\cut{A_1}$ in $G_1$, and hence by the argument above, contributes at most $\gamma\, \capa\cut{A_\ell}$.  Moreover, if $\dem_{\ell-1}^e$
  and $\capa_{\ell-1}^e$ denotes the total demand and capacity cut by
  $A_\ell$ inside $G_{\ell-1}^e$, then by the induction hypothesis, we
  have $\dem_{\ell-1}^e \leq (\ell-1)\cdot\gamma \cdot \capa_{\ell-1}^e$.
  Since these demands and capacities are coming from disjoint sets of edges, we can add these inequalities to conclude that
  $\dem(A_{\ell},\Abar_{\ell})$ is at most $\gamma\,\capa\cut{A_\ell} + \sum_{e}
  \dem_{\ell-1}^e \leq (\gamma + (\ell-1)\cdot\gamma)\cdot
  \capa(A_{\ell},\Abar_{\ell})$. The second part of the claim
  follows.

  Finally, let $(A_\ell,\Abar_\ell)$ be an inadmissible cut with $s,t
  \not\in A_{\ell}$; by \lref[Lemma]{lem:os}, we can assume it is
  connected. Let $\cut{A_1}$ denote the projection of $\cut{A_{\ell}}$
  onto $G_1$. Our construction guarantees that $A_1$ is either empty, or
  induces a connected cut in $G_1$. In the former case, $A_{\ell}$
  induces an inadmissible cut in some $G_{\ell-1}^e$; the demand and
  capacity cut by $A_{\ell}$ equals those in this inadmissible cut of
  $G_{\ell-1}^e$, so we can use the inductive hypothesis. Since
  both the demands and capacities are scaled by the same amount, this
  gives us the proof for the first case. In the latter case, for every
  edge $e \in G_1$ cut by $\cut{A_1}$, the cut $\cut{A_\ell}$ induces an
  admissible cut in $G_{\ell-1}^e$. Let $\capa_{\ell-1}^e$ and
  $\dem_{\ell-1}^e$ denote the capacity and demand cut by $A_\ell$
  inside $G_{\ell-1}^e$. By the inductive hypothesis, each of these
  admissible cuts has at least unit capacity in the unscaled version of
  $G_{\ell-1}$, so the scaled-down capacity $\capa_{\ell-1}^e \geq
  \capa^{G_1}_e$. Summing over all cut edges, $\capa\cut{A_{\ell}} \geq
  \capa^{G_1}\cut{A_1}$.  The level-$\ell$ demand cut by $A_\ell$ is
  equal to the demand cut in $G_1$ by $A_1$; by the last assumption this
  is at most $\capa^{G_1}\cut{A_1} \leq \capa\cut{A_\ell}$. Moreover,
  using the second part of the induction hypothesis, we conclude that
  for any $e$, $\dem_{\ell-1}^e \leq (\ell-1)\,\gamma\,
  \capa_{\ell-1}^e$. Since the demands and capacities contribution to
  $\dem_{\ell-1}^e$ and $\capa_{\ell-1}^e$ are disjoint for different
  edges $e$, we can add these inequalities to conclude that the total
  demand cut is at most $\capa\cut{A_\ell} + \sum_e \dem_{\ell-1}^e \leq
  \capa(A_\ell,\Abar_{\ell})( 1+ (\ell-1)\gamma)$, which completes the
  proof.
\end{proof}

\subsubsection{Putting It Together}
\label{sec:endgame-NP}

Let $G_1$ be the instance defined in \lref[Section]{sec:basic-bb} and
$G_\ell$ be obtained by the powering operation starting with $G_1$.
\lref[Lemmas]{lem:buildingblock} and~\ref{lem:powercompleteness} imply
that if $H$ has a cut of size $cm$, then $G_\ell$ has a cut of sparsity
$\frac{1}{\ell c}$. Moreover, using \lref[Lemma]{lem:powersoundness} along
with \lref[Lemma]{lem:buildingblock} shows that if $H$ has max cut size at
most $sm$, then the sparsest cut in $G_\ell$ has sparsity at least
$\frac{1}{1+(\ell-1)s}$. Hence, a $cm$-vs-$sm$ hardness for \mc
translates to a $\frac{\ell c}{1+(\ell - 1)s} \geq \frac{c}{s}(1 -
\frac{1-s}{\ell s})$ hardness for Sparsest Cut. Taking $\ell =
\frac{1-s}{s\eps} = \Omega(1/\eps)$ gives us
a hardness of $\frac{c}{s}(1-\eps)$.

Note that Lee and Raghavendra~\cite{LR10} show the integrality gap of
the natural LP relaxation for \nusc is $2$ for series-parallel graphs;
Chekuri, Shepherd, and Weibel~\cite{CSW10} give a different analysis of
the integrality gap lower bound. Their instances are the graphs $G_\ell$
above, but with $K_{n}$ as the \mc instance $H$. In hindsight,
their gaps follow from the fact that the integrality gap of the LP
relaxation of \mc on $K_n$ is 2. This theme will be revisited when we
show an integrality gap for the \SA LP using the \SA integrality gaps
for \mc.

%%% Local Variables: 
%%% mode: latex
%%% TeX-master: "sp-cut"
%%% End: 

\ifstoc
\section{A Tight UG Hardness}
\else
\section{A Tight Unique Games Hardness}
\fi
\label{sec:ug-hardness}

In this section, we show that, assuming the Unique Games Conjecture, the
Sparsest Cut problem is hard to approximate better than a factor of $2$,
even on bounded treewidth graphs. Specifically, for every constant $\eps
> 0$, having a polynomial-time
algorithm for every fixed value of treewidth that gave a $(2-\eps)$-approximation to \spcut would
violate the Unique Games Conjecture.

The proof in this section first abstracts out a useful form of the
Unique Games problem and builds a basic instance from it that shows a
hardness of $3/2 - \eps$. Then we use the powering (``fractalization'')
operation from \lref[Section]{sec:recursive} to boost the hardness to $2 -
\eps$. 

\subsection{A Convenient Form of Unique Games}
\label{sec:forms-of-UG}

One standard form of the \emph{Unique Label Cover} (a.k.a. Unique Games)
problem is the following. We are given a bipartite graph $B =
(U,V,E_B)$. There is a label set with $d$ labels. Each edge $(u,v) \in
E_B$ has an associated bijective map $\sigma_{u,v}: [d] \to [d]$. A
labeling is a map from $U \cup V$ to $[d]$, and \emph{satisfies} an edge
$(u,v) \in E_B$ if
\[ \sigma_{u,v}(\text{label}(u)) = \text{label}(v). \]
The optimum of the Unique Label Cover problem is the maximum fraction of
edges satisfied by any labeling.

\begin{conjecture}[Unique Games Conjecture]
  For any $\eta, \gamma > 0$, there is a large enough constant $d =
  d(\eta, \gamma)$ such that it is NP-hard to distinguish whether a
  Unique Label Cover instance with label size $d$ has optimum at least
  $1-\eta$ or at most $\gamma$. 
\end{conjecture}

It will be most convenient for us to consider non-bipartite versions of
\ulc, where there is a general (multi)-graph $H = (V_H,E_H)$. For each
$e = (v,w) \in E_H$ there is again a bijective map $\sigma_e: [d] \to
[d]$ and the goal is to find a labeling maximizing the number of
satisfied edges. We call such a multigraph $H$ a \emph{union of cliques} if
there exists a partition of $E_H$ into (edge-disjoint) cliques $C_1, C_2,
\ldots, C_r$ for some $r$, i.e., each $C_i$ is a complete graph on some
subset $S_i \subseteq V_H$. (Recall that $H$ is a multigraph, so these
sets $S_i$ may have more than single vertices in common, resulting in
parallel edges.) We call a
\ulc instance $(H, \{\sigma_e\}_{e \in E_H}, d)$  {\em $\Delta$-\awesome}
if
\begin{OneLiners}
\item The edge set $E_H$ is an (edge-disjoint) union of cliques $C_1$, $C_2$,
  $\ldots$, $C_n$ where $n := |V_H|$.
  \item Each clique $C_i$ is over some subset $S_i$ of size $\Delta$.
  \item Each vertex $v \in V_H$ lies in exactly $\Delta$ cliques.
\end{OneLiners}
Note these properties mean that the degree of each vertex is exactly
$\Delta(\Delta -1)$, where we count parallel edges. Moreover, the total
number of edges in $H$ is $n \binom{\Delta}{2}$. A \ulc instance is
{\em \awesome} if it is $\Delta$-\awesome for some $\Delta$. (The use of such
a \ulc instance is also implicit, e.g., in~\cite{KKMO}.)

\begin{lemma}
\label{lem:awesomeinstance}
  Assuming the UGC, for any $\eta > 0$, there is a large enough constant
  $d = d(\eta)$ such that it is NP-hard to distinguish whether a
  \awesome \ulc instance with label size $d$ has optimum at least
  $1-\eta$, or at most $\eta$. Moreover, this holds for
  $\Delta$-\awesome instances where $\eta < 1/2\Delta$.
\end{lemma}
\stocoption{We defer the proof to the full version.}{
\begin{proof}
  We can assume we have instances $B = (U,V,E_B)$ of bipartite \ulc
  which are $\Delta$-regular---all vertices in $U \cup V$ have
  degree $\Delta$ in $B$---where we cannot distinguish between $(1 -
  \delta)$-versus-$\delta$ fraction of satisfiable constraints. We
  define a non-bipartite \ulc instance $H = (V,E_H)$ on the set $V$ as
  follows: For each $u \in U$, add edges in $E_H$ between all its
  neighbors. (Hence, the number of edges added between $v, w \in V$
  equals the number of common neighbors they have.) Since each vertex $u
  \in U$ results in a clique being added among its $\Delta$ neighbors,
  the new instance is a union of $n$ edge-disjoint $\Delta$-cliques.
  Moreover, each vertex in $V$ belongs to $\Delta$ cliques, equal to its
  degree in $B$. So the instance is $\Delta$-\awesome.  Finally, the
  bijection constraints are $\sigma'_{vw} := \sigma_{uw} \circ
  \sigma_{uv}^{-1}$, with label set~$[d]$.  
  
  Suppose the bipartite instance $B$ was $1 - \delta$ satisfiable. Let
  $\textsf{label}: U\cup V \rightarrow [d]$ denote a labeling that
  satisfies $(1-\delta)$ fraction of the constraints. Let  $\delta_u$
  denote the fraction of constraints incident on $u$ that are violated
  by $\textsf{label}$ so that $\sum_u \delta_u \leq n\delta$. In the
  clique corresponding to $u$, this labeling violates at most
  $(\delta_u\Delta) \cdot (\Delta -1)$ constraints . Thus the total
  number of violated constraints is at most  $\sum_u \delta_u\Delta
  (\Delta -1) \leq n\delta\Delta(\Delta-1) = 2\delta n{\Delta \choose
    2}$. Thus at least $1 - 2\delta$ of the constraints in $H$ are
  satisfiable.

  Conversely, suppose a $\eta$ fraction of constraints in the union of
  cliques instance $H$ was satisfied by a labeling $\chi$. We would like
  to extend this coloring to the vertices of $U$ so that at least
  an $\eta$ fraction of the constraints in $B$ are satisfied. To this end, consider any
  $u \in U$ and restrict attention to the clique of edges added among
  the neighbors of $u$. Suppose $\eta_u$ fraction of the
  $\binom{\Delta}{2}$ edges in this clique were satisfied in $H$ by
  $\chi$. Note that $\E_u[ \eta_u ] = \eta$. Suppose these $\eta_u
  \binom{\Delta}{2}$ satisfied edges form $t$ connected components, with
  sizes $k_1 \geq k_2 \geq \ldots \geq k_t$. The maximum number of
  satisfied edges within these components is $\sum_{i = 1}^t
  \binom{k_i}{2}$, which must be at least $\eta_u \binom{\Delta}{2}$.
  This implies that $\frac{\Delta}{k_1} \binom{k_1}{2} \geq \eta_u
  \binom{\Delta}{2}$, and in turn, $k_1 \geq \eta_u \cdot \Delta$. In
  other words, there exists a connected component of satisfied edges
  with $\eta_u \Delta$ vertices. Pick any vertex $v$ in this component and
  set $\chi(u) := \sigma^{-1}_{uv}(\chi(v))$.  Note that this label for
  $u$ satisfies not only the edge $(u,v) \in E_B$, but also the edges
  $(u,w)$ for $w$ lying in this connected component. (This follows by
  the uniqueness of the constraints.) Do this independently for each $u
  \in U$. The total fraction of constraints satisfied by this extension
  of $\chi$ to a labeling of $U \cup V$ is now $\E_u[ \eta_u ] = \eta$,
  which completes the proof.

  Finally, observe that in any $\Delta$-\awesome instance, we can
  find a matching of size at least $1/2\Delta$ and satisfy all these
  edges; hence the parameter $\eta$ must be below $1/2\Delta$.
\end{proof}
}
In the next section, we give a reduction from such $\Delta$-\awesome
instances of \ulc to \spcut on constant treewidth graphs. 
A little notation will be useful: for a vertex $x = (x_1, x_2, \ldots,
x_d) \in \bits^d$ and a permutation $\sigma$, define $\sigma(x)$ to be
the vector
\[ (x_{\sigma^{-1}(1)}, x_{\sigma^{-1}(2)}, \ldots, x_{\sigma^{-1}(d)}) \]
and define $\overline{x} = (-x_1, -x_2, \ldots, -x_d)$.

\subsection{The Basic Instance}
\label{sec:basic-ulc}

Consider a $\Delta$-\awesome \ulc instance $H = (V_H, E_H)$ with
bijections $\{\sigma_{u,v}\}_{(u,v) \in E_H}$ and label set $d$. Let $n
:= |V_H|$ be the number of vertices in $H$ and $m := |E_H| = n
\binom{\Delta}{2}$ be the number of edges in $H$. Let $N := n\cdot 2^d$
and $M := m \cdot 2^d$; hence $M/N = m/n = \binom{\Delta}{2}$.

The \spcut instance $(G,D)$ on the new set of nodes $V$ is the following:
\begin{itemize}
\item \textbf{Nodes.} Take a cube $Q_v = \bits^d$ for each $v\in
  V_H$. Add two new nodes $s, t$. This is the new set of nodes $V$ of
  size $n\cdot2^d + 2 = N+2$. We will use $x \in \bits^d$ to denote cube
  nodes (i.e., those in $V \setminus \{s,t\}$, and will use the notation
  $x^{v}$ to indicate that $x$ is a cube node in cube $Q_v$.

\item \textbf{Supply Edges.} Add edges of capacity $1/N$ from $s$ to
  each node in $\cup_{v \in V_H} Q_v$, and the same from $t$ to these
  nodes. There are $2N$ such edges, which we call \emph{star edges}.

  Add all the hypercube edges, but with capacity $\alpha/N$ for some
  parameter $\alpha > 0$.  There are $n \cdot d2^{d-1} = dN/2$
  \emph{cube edges}, which have total capacity $\alpha d/2$. (Think of
  $\alpha$ as a small quantity.)

\item \textbf{Demand Edges.} For each edge $e = (v,w) \in E_H$ and for
  each $x \in \bits^d$, add a demand edge with demand $1/M$ between $x
  \in Q_v$ and $\overline{\sigma_{vw}(x)} \in Q_{w}$. (If
  $\sigma_{vw}$ were the identity permutation, then we would add a
  demand edge between each node in $Q_v$ and its antipodal node in
  $Q_{w}$.) Thus there is a total demand of $2^d/M$ for each edge $e
  \in E_H$ and hence a total of $m\cdot 2^d/M = 1$ of such demand.   Observe that there are
  $\Delta(\Delta-1)$ demand edges $(x^{v}, \cdot)$ incident to
  each node $x^{v} \in Q_v$ for each $v \in V_H$.

%  Finally, add a demand edge between $s$ and $t$, of demand value
%  $\beta = 1/2$. 
  
  % % We define $\beta := 1/2 + \eta$, where $\eta$ is the
  % hardness parameter for the \ulc instance. 
  %Let $\beta = 1/2$.
\end{itemize}

This instance will be a good starting point for the powering operation
of Section~\ref{sec:recursive}. Recall that a cut is admissible if it
separates $(s,t)$, and is inadmissible otherwise. We will show that in
the good case, there is a sparse admissible cut, whereas in the bad
case, all admissible cuts have much higher sparsity. Additionally, we
will prove a (weaker) lower bound on the sparsity of inadmissible cuts.

Intuitively, the factor of two comes from the following facts: if we connect two hypercubes with demands between $x$ in the first hypercube to $\bar{x}$ in the second, choosing the same dictator cut on both hypercubes cuts {\em all} demand pairs, while choosing different dictator cuts on the the two hypercubes cuts only half the demand pairs. Thus restricted to dictator cuts, we get a gap of two. How do we exclude cuts that are far from dictators\footnote{This is where we are able to improve on the \mc UG-hardness of~\cite{KKMO,MOO}, where the noise operator is applied to such a construction, and majority-is-stablest is used to exclude cuts far from dictators.}? Here we use the fact that sparse cuts should cut few {\em supply edges} within each hypercube. Freidgut's junta theorem tells us that cuts that do not cut much more than a dictator cut in a hypercube, are close to juntas, which is sufficient to be able to ``decode'' a sparse cut into a good assignment to the unique games instance.

We start by recording some basic properties of $G$, which will prove
useful for the powering operation.
\begin{observation}
Let $G$ be the network defined above.
\begin{OneLiners}
\item The total capacity in the network is $(1+ \alpha d/2)$,  and the total demand is $1$. 
\item The treewidth of the supply graph is at most $2^d$, the size of each
hypercube. 
\item Any admissible cut has capacity at least $1$.
\end{OneLiners}
\label{obs:capone}
\end{observation}

%Observe the total capacity in the network is $(1+ \alpha d/2)$,
%whereas the total demand is $(1+\beta)$. Moreover, observe that the
%treewidth of the supply graph is at most $2^d$, the size of each
%hypercube. 

\subsubsection{Completeness}
\label{sec:complet-ulc}

\begin{lemma}
  \label{lem:comp}
  If there exists a labeling $f: V \to [d]$ satisfying at least $1 -
  \eta$ fraction of the \ulc instance, then there exists a cut in
  $(G,D)$ of capacity at most $(1+\alpha/2)$ that cuts at least $(1-\eta)$ demand.
%  sparsity at most 
%  \[ \Phi^\star := \frac{(1+\alpha/2)}{3/2 - \eta}. \]
\end{lemma}

\begin{proof}
  Consider the set $A$ that contains $s$ and, for each $v \in V_H$,
  contains $\{ x \in Q_v \mid x_{f(v)} = 1 \}$. This cut separates each
  node in $V \setminus \{s,t\}$ from either $s$ or $t$, and hence
  separates $N$ of the star edges to cut total capacity $N \cdot 1/N
  = 1$.  Moreover, it cuts exactly $2^{d-1}$ edges from each hypercube,
  and hence $\alpha/N \cdot n2^{d-1} = \alpha/2$ capacity in the cube
  edges.  This gives a total of $(1+\alpha/2)$ capacity cut by $(A,
  \Abar)$.

 There are at  least $(1-\eta)m$ edges $(v,w) \in E_H$ with $\sigma_{vw}(f(v))
  = f(w)$. Consider one such edge $(v,w)$ and some $x \in Q_v$; say
  $x_{f(v)} = b \in \bits$. Then its demand edge corresponding to
  $(v,w)$ goes to $\overline{\sigma_{vw}(x)} \in Q_{w}$, whose
  $f(w)^{th}$ coordinate contains
  \[ - x_{\sigma^{-1}_{vw}(f(w))} = -x_{f(v)} = -b. \] Hence, $A$
  cannot contain both $x \in Q_v$ and $\overline{\sigma_{vw}(x)} \in
  Q_{w}$ and thus cuts all $2^d$ demands corresponding to
  $(v,w) \in E_H$. The total demand cut is therefore at least
%  \[ \beta + (1-\eta)m\cdot 2^d/M = (\beta + 1 - \eta) = (3/2 - \eta). \]
   \[ (1-\eta)m\cdot 2^d/M = (1 - \eta). \]
  This completes the proof.
\end{proof}

\subsubsection{Soundness}
\label{sec:sound-ulc}

\begin{theorem}
  \label{thm:sound}
  For $\eps > 0$, let $\alpha = \eps^2$. If the
  resulting Sparsest Cut instance has an admissible cut $(A,\Abar)$ in $(G,D)$ of
  sparsity less than $(2 - 4\eps)$, then there exists a labeling for the
  \ulc instance that satisfies at least $\eta' := (\eps/2 -
  1/\Delta)/2^{O(\eps^{-4})}$ fraction of the edges in $H$.
\end{theorem}

\begin{proof}
  A road-map of the proof: We first %show that such a cut $(A, \Abar)$
%  must be $s$-$t$-separating, and must separate a lot of the cross
 % demand. We also 
show that a large fraction of the cubes are
  ``good'', i.e., for most of the cubes, very few of the supply edges
  are cut.  Using Freidgut's Junta Theorem, these good cubes are close
  to being ``juntas'', i.e., each contains a small number of influential
  coordinates.  Finally, for a large amount of the demand to be
  separated using these juntas, there is a non-trivial number of edges in $H$
  that have end-points whose juntas share a common coordinate, so a
  randomized rounding procedure gives the claimed good labeling for $H$.

  \begin{claim}
    \label{clm:most-cross}
    The cut $(A,\Abar)$ must separate more than $(\frac12 + \eps)$ demand.
  \end{claim}

  \begin{proof}
    By \lref[Observation]{obs:capone}, $(A,\Abar)$ has capacity at least $1$. For the sparsity to be less than $(2 - 4\eps)$, $(A, \Abar)$  must cut at least $1/(2 - 4\eps) > (\frac12 + \eps)$ units of demand. 
%    By \lref[Claim]{clm:sepa-cut}, $(A,\Abar)$ separates $s$ from $t$.
%    Hence it cuts at least $N \times 1/N = 1$ unit of demand on the star
%    edges. For the sparsity to be less than $(1 - \eps)$, $(A, \Abar)$
%    must cut at least $1/(1 - \eps) > (1 + \eps)$ units of demand. But
%    only $\beta$ units of this can be the $s$-$t$ demand, the rest $(1 +
%    \eps - \beta)$ must come from the cross demand. Using $\beta = 1/2$
%    gives the second bound.
  \end{proof}

  Call a vertex $v \in V_H$ \emph{good} if the number of cube edges from
  $Q_v$ cut by $(A, \Abar)$ is at most $\frac{32}{\alpha \eps} 2^{d-1}$;
  we also say such a cube is good. (Note that a dimension cut---a.k.a.\
  a dictator---would cut exactly $2^{d-1}$ edges, so good cubes do not
  have ``too many more'' cut edges than a dictator.)
  \begin{claim}
    \label{clm:many-good}
    There are at most $(\eps/8)n$ bad vertices in $V_H$.
  \end{claim}
  \begin{proof}
    Suppose not: Then the number of cube edges cut is strictly more than
    $(\eps/8)n \cdot (32/\alpha \eps) 2^{d-1}$, each of capacity
    $\alpha/N$. This gives a total capacity of more than $2$.
    But the total demand in the instance is $1$, which
    would imply the sparsity of $(A, \Abar)$ is more than $2$, a
    contradiction.
  \end{proof}

  The rest of the argument shows that a cut separating $(1/2 + \eps)$
  units of demand (by \lref[Claim]{clm:most-cross}) that is good on
  most cubes (by \lref[Claim]{clm:many-good}) can be decoded to a labeling
  for the \ulc instance. (Henceforth, we will not worry about the supply
  edges, etc., and will just consider the cubes and the demands.)

  Let us view the cut $(A, \Abar)$ as a function $f'': \cup_{v \in V_H}
  Q_v \to \bits$. We now perform two small changes to this function.
  Firstly, for any bad vertex $v \in V_H$, let us change $f''$ on the
  nodes of $Q_v$ to be identically equal to $1$. Note that the new
  function (which we call $f'$) might separate less demand, but
  since we changed the function value on an $\eps/8$ fraction of the
  cube nodes and the instance is regular, the demand separated
  must still be $(1/2 + \eps - 2\cdot\eps/8) = (1/2 + 3\eps/4)$.

  Secondly, we use the following rephrasing of Freidgut's Junta
  Theorem~\cite{Freidgut98} that has been used in the context of hardness
  for cut-related problems (see, e.g.,~\cite{CKKRS05}):
  \begin{theorem}[Freidgut's Junta Theorem]
    \label{thm:junta}
    If a function \ifstoc \\ \fi $f':\bits^d \to \bits$ cuts at most $c 2^{d-1}$ edges
    of the cube $\bits^d$, then it is $\eps$-close (i.e., differs in at
    most $\eps2^d$ points) from a function $f: \bits^d \to \bits$ that
    is a $\exp\{O(c/\eps)\}$-junta.
  \end{theorem}
  (Recall that the function $f$ is a \emph{$K$-junta} if there is a set
  $S \sse [d]$ of size at most $K$, such that for any $x \in \bits^d$,
  specifying the coordinates of $x$ that lie within $S$ determines the
  value $f(x)$. I.e., $f$ is constant on the subcubes obtained by fixing
  the bits in $S$ and running over all the other bits.)

  Using \lref[Theorem]{thm:junta} for each good cube, we replace the
  function $f'$ restricted to that cube by its $\eps/8$-close $K$-junta,
  where $K := \exp\{ O(1/\eps^2 \alpha) \} = \exp\{ O(\eps^{-4}) \}$.
  Here we used the fact that the number of edges cut in a good cube is at most
  $\frac{32}{\alpha \eps} 2^{d-1}$ the assumption $\alpha =
  \eps^2$. This redefined function, which we call $f$, differs from $f'$
  on at most $\eps/8$ fraction of all the $n\cdot 2^d$ cube nodes,
  which might result in at most $2\cdot \eps/8 = \eps/4$  fraction of the 
  demand no longer being separated. This still leaves at least $1/2 +
  3\eps/4 - \eps/4 = 1/2 + \eps/2$ demand separated. Moreover, now
  the newly redefined function $f$ is a $K$-junta on each of the
  cubes---on the bad cubes by the first redefinition and on the good
  cubes by the second---and $f$ still separates $(1/2 + \eps/2)$ units of
  demand.

  For each $v \in V_H$, define the set $J_v \sse [d]$ to be the set of $K$
  most influential coordinates of $f$ restricted to the cube $Q_v$:
  since $f|_{Q_v}$ is a $K$-junta, specifying $\{x^v_i\}_{i \in S}$
  fixes the value $f(x^v)$. Secondly, call an edge $(v, w) \in E_H$
  \emph{compatible} if $\sigma_{vw}(J_v) \cap J_{w} \neq \emptyset$, i.e.,
  if $f|_{Q_v}$ and $f|_{Q_{w}}$ ``share'' an influential coordinate
  (after applying the right permutation). Observe that if an edge is
  compatible, then assigning each vertex $v$ a label randomly chosen from
  its set $J_v$ would satisfy each compatible edge with probability at
  least $1/K^2$. The following lemma is now the final piece in the
  argument.

  \begin{lemma}
    \label{lem:compatible}
    The number of compatible edges $(v,w) \in E_H$ is at least $(\eps/2
    - 1/\Delta)m$.
  \end{lemma}

  \begin{proof}
    Consider all the cubes $Q_v$. Recall that we don't care about the supply
    edges for this argument, only the demand edges. For each
    $(v,w) \in E_H$, these demand edges give a matching between
    the nodes of $Q_v$ and $Q_{w}$; the union of all these matchings
    gives us all the demand edges.  Finally, $f$ gives us a
    $\bits$-coloring of the nodes of the cubes and separates $(1/2 +
    \eps/2)$ fraction of these demand edges.

    For each cube $Q_v$, collapse all the nodes $x^v$ whose values on
    the coordinates in $J_v$ are the same to get a cube $\Qhat_v$
    isomorphic to $\bits^K$. (Each new node $\xhat^v$ in $\Qhat_v$
    comprises of $2^{d-K}$ nodes collapsed together.) Moreover, since
    any two nodes collapsed together agree on their $f$-value, the
    number of separated demand edges in the resulting multigraph
    remains unchanged.

    Suppose $(v,w)$ is not compatible. We claim that the $2^d$ demand
    edges between $Q_v$ and $Q_w$ now form a complete bipartite
    multigraph between $\Qhat_v$ and $\Qhat_w$, with $2^{d-2K}$ demand
    edges going between each $\xhat^v \in \Qhat_v$ and $\xhat^w \in
    \Qhat_w$.  For simplicity, assume that $\sigma_{vw}$ is
    the identity map, so incompatibility means that $J_v \cap J_w
    = \emptyset$, and imagine that $J_v = \{1,2, \ldots, K\}$ and $J_w = \{K+1,
    K+2, \ldots, 2K\}$. A demand edge that used to go between $(b_1,
    b_2, \ldots, b_d) \in Q_v$ and $(-b_1, -b_2, \ldots, -b_d) \in Q_w$
    now goes between $(b_1, b_2, \ldots, b_K) \in \Qhat_v$ and
    $(-b_{K+1}, \ldots, -b_{2K}) \in \Qhat_w$. There are $2^{d-2K}$ of
    these edges, and they are the only ones. This proves the claim.

    Now suppose there were no compatible edges at all. Then for each
    edge $(v,w)$ in $H$, we would add a complete bipartite multigraph
    between the $2^K$ nodes in $\Qhat_v$ and $\Qhat_w$. This would be
    exactly the multigraph $H \times I_{2^K}$ (where $I_{\ell}$ is a
    graph on $\ell$ vertices and no edges at all), but with each edge
    replicated some number $2^{d-2K}$ number of times. It is easy to
    show that, because $H$ is a union of $\Delta$-cliques, $f$ can cut
    at most $1/2 + 1/\Delta$ fraction of the demand edges; note this is
    precisely where we use the union-of-cliques structure of $H$. The
    simple proof is \stocoption{omitted}{in \lref[Claim]{clm:clique-cut}}.

    However, we claimed that $f$ cut $1/2 + \eps/2$ fraction of the edges, so
    there must be at least one compatible edge. Observe that the
    number of cut edges behaves in a Lipschitz fashion as we change the
    number of compatible edges: Each compatible edge $(v,w) \in E_H$
    means the cut may be higher by at most an additive $2^d$ amount
    (since this is the total number of demand edges corresponding
    to an $H$-edge), which is a $1/m$ fraction of the total
   demand. If we want the fraction of demand edges cut to
    increase by $\eps/2 - 1/\Delta$, and each compatible edge can
    increase this additively by at most $1/m$, we have at least $(\eps/2
    - 1/\Delta)m$ compatible edges.
  \end{proof}

  Choosing a label for each $v \in V_H$ randomly from $J_v$ and
  using \lref[Lemma]{lem:compatible}, the expected fraction of satisfied
  edges in $E_H$ is at least $(\eps/2 - 1/\Delta)/K^2$. Noting that $K =
  \exp\{ O(\eps^{-4}) \}$ completes the proof of \lref[Theorem]{thm:sound}.
\end{proof}

% By taking the number of labels $d$ large enough in the \ulc instance, we
% can make $\eta > \eta'$. (Note that while the parameter $\Delta$ appears
% in the expression for $\eta'$, the value $1/\Delta$ decreases as the
% number of labels $d$ increases, which means we can indeed ensure that
% $\eta > \eta'$ for large enough $d$.) 

%Consider setting $\eps$ such that $\eta < \eta'$ (and hence $\eta \ll
%\eps$) and $1/\Delta < \eps/4$; for such a setting, a hardness for \ulc of
%$1- \eta$ versus $\eta$ translates to differentiating between having a
%sparsest cut of value $\frac{1+\alpha/2}{3/2 - \eta} \leq
%\frac{1+\eps^2/2}{3/2(1 - \eps)} \leq 2/3(1+2\eps)$ (assuming that
%$\eps$, and hence $\eta$ are small enough), versus that of value
%$1-\eps$, all on graphs of treewidth at most $2^d$.

Finally, we show that inadmissible cuts cannot have sparsity better than $1$
  \begin{claim}
    \label{clm:inadmissiblecutsok}
Any inadmissible cut $\cut{A}$ satisfies $\dem \cut{A} \leq \capa \cut{A}$.   
  \end{claim}

  \begin{proof}
    By \lref[Lemma]{lem:os}, we can assume $(A,\Abar)$ is connected and assume w.l.o.g. that $s, t \notin A$. Then $A$ is a connected subset of one cube $Q_v$. This means
    $\partial A$ contains all star edges adjacent to $A$, so $\capa(A,\Abar) \geq
    2|A|\cdot 1/N$.  Moreover, for each $x^v \in A$, all the
    $\Delta(\Delta-1)$ demand edges incident to $x^v$ are cut, so the
    total demand cut is $\Delta(\Delta-1) \cdot |A| \cdot 1/M = 2|A|/N$.
  \end{proof}

\stocoption{}{
\subsubsection{\mc on Lifts of $K_n$}

The following simple lemma was useful in the analysis of Theorem~\ref{thm:sound}.
\begin{claim}
  Let $H$ be a multigraph that is a union of cliques, each on exactly
  $\Delta \geq 2$ vertices. Let $I_{\ell} = (\{1, 2, \ldots, \ell\},
  \emptyset)$ be a graph of $\ell$ vertices with no edges. For each
  $\ell \in \Z_{\geq 1}$, every cut in the graph $H \times I_\ell$
  contains at most $\frac12 + \frac{1}{2(\Delta -1)} \leq \frac12 +
  \frac1\Delta$ fraction of the edges.
  \label{clm:clique-cut}
\end{claim}

\begin{proof}
  For $\ell = 1$, $H \times I_\ell = H$. Consider any $\Delta$-clique in
  $H$; any cut (i.e., a two-coloring of the vertices) separates at most
  \[ \left\lceil\frac{\Delta}{2}\right\rceil \left\lfloor\frac{\Delta}{2}\right\rfloor \leq
  \binom{\Delta}{2} \left(\frac12 + \frac{1}{2(\Delta -1)} \right) \]
  edges. Hence, any cut separates at most the claimed fraction of edges
  in each clique, and the fact that $H$ is a union of (edge-disjoint) cliques
  gives the result for $\ell = 1$.

  For general $\ell$ case: 
  For each $v$, let $b_v$ be the fraction of its $\ell$ copies colored
  blue.  Color $v$ blue with probability \ $b_v$ and red with probability\ $1-b_v$. This gives a
  coloring (i.e., cut) of the vertices of $H$. For any edge $(u,v)$ in $H$,
  this coloring cuts $\ell^2 \cdot (b_u (1-b_v) + (1-b_u)b_v)$
  edges in $H \times I_\ell$, which is exactly $\ell^2$ times the
  probability of $(u,v)$ being cut by the random coloring in $H$. Hence
  the fraction of edges cut in $H \times I_\ell$ is exactly the expected
  fraction of edges in $H$ cut by the random coloring. But the argument
  above implies that each $2$-coloring of $H$ cuts at most $1/2 +
  1/2(\Delta-1)$ fraction of its edges, which lower bounds the
  expectation and proves the claim.
\end{proof}
}

\subsection{Putting It Together}
\label{sec:endgame-ulc}

To boost the hardness result from \lref[Section]{sec:basic-ulc}, we will
use the powering operation on the instance defined above. %formed by {\em removing} the
%$s$-$t$ demand from the instance. (This mimics what we did in
%\lref[Section]{sec:recursive} for the reduction from max-cut.)
Intuitively, our proof gave us a $2-\eps$ gap as we long as we could ignore the inadmissible cuts.  The powering operation decreases the sparsity in both the good and the bad cases. However, it ensures that even inadmissible cuts induce admissible cuts at lower levels and thus cannot have much better sparsity.

Let $G_1$ be the instance from \lref[Section]{sec:basic-ulc}. Then \lref[Observation]{obs:capone}, \lref[Theorem]{thm:sound}, and \lref[Claim]{clm:inadmissiblecutsok} can be summarized as:
\begin{lemma}
Suppose that there is no labeling $f: V \rightarrow [d]$ satisfying an
$\eta$ fraction of the unique label cover instance. Then there exists
$\eta' := \eta'(\eps, \eta)$ such that the graph $G_1$ satisfies:
 \begin{OneLiners}
  \item Any admissible cut $(A,\Abar)$ has capacity $\capa(A,\Abar)$ at
    least $1$.
  \item Any admissible cut $(A,\Abar)$ cuts at most $(\frac 1 2 +\eta\rq{}) \cdot
    \capa(A,\Abar)$ demand.
  \item Any inadmissible cut $(A,\Abar)$ cuts at most $\capa(A,\Abar)$ demand.
 \end{OneLiners}
  \end{lemma}

Let $G_\ell$ be the instance applying the powering operation to the
instance $G_1$, and let $\ell=4/\eps$ and $\alpha= \eps^2$. Using \lref[Lemma]{lem:comp} and \lref[Lemma]{lem:powercompleteness}, a straightforward calculation shows:
\begin{lemma}
If there exists a labeling $f: V \rightarrow [d]$ satisfying at least $1-\eta$ fraction of the unique label cover instance, then there exists a cut $\cut{A_\ell}$ in $G_\ell$ such that $\capa \cut{A} \leq (1+\frac{\alpha}{2})^{\ell} \leq 1+4\eps$ and $\dem \cut{A} \geq \ell(1 - \eta) $.
\end{lemma}
Hence, in the ``yes'' case, the sparsity is at most
\begin{gather}
  \frac{(1+4\eps)}{\ell(1-\eta)}. \label{eq:8}
\end{gather}
On the other hand, \lref[Lemma]{lem:powersoundness} implies that:
\begin{lemma}
Suppose that there is no labeling $f: V \rightarrow [d]$ satisfying an $\eta$ fraction of the unique label cover instance. Then the graph $G_\ell$ satisfies:
 \begin{OneLiners}
  \item Any admissible cut $(A,\Abar)$ cuts at most $\ell(\frac 1 2 +\eta\rq{}) \cdot
    \capa(A,\Abar)$ demand.
  \item Any inadmissible cut $(A,\Abar)$ cuts at most $(\ell(\frac 1 2 +\eta\rq{}) + \frac 1 2)\capa(A,\Abar)$ demand.
 \end{OneLiners}
 \end{lemma}
In both these cases, the sparsity is at least
\begin{gather}
  \frac{1}{\ell(\frac 1 2 +\eta\rq{}) + \frac{1}{2}} \geq
  \frac{1}{\ell(\frac 1 2 + \eta\rq{} + \eps/2)}. \label{eq:9}
\end{gather}
From~(\ref{eq:8}) and~(\ref{eq:9}), we conclude that the hardness is at
least $\frac{2(1-\eta)}{(1+4\eps)(1+2\eta\rq{}+\eps)} \geq 2
(1-5\eps-\eta-2\eta\rq{})$. Thus for any $\eps\rq{}>0$, we can pick $\eps$ small enough to get a $2-\eps\rq{}$ hardness.
%%% Local Variables: 
%%% mode: latex
%%% TeX-master: "sp-cut"
%%% End: 

\ifstoc
\else
\ifstoc
\section{SA Integrality gaps}
\else
\section{A $2-\eps$ Integrality Gap for \SA}
\fi
\label{sec:SA-gaps}

In this section, we show that how to translate \SA integrality gaps for
the \mc problem into corresponding \SA integrality gaps for \nusc. 

\subsection{The \SA LP and Consistent Local Distributions}
\label{sec:sa-props}

We begin by recording a standard result stating that an $r$-round \SA
solution is essentially equivalent to the existence of a collection of consistent
``local'' distributions $\{\D_S\}_{|S| \leq r}$.

\begin{theorem} \label{thm:SAiffDist} There exists a $(x,y) \in
  \mathbf{SA}_r(n)$ if and only if for every $S \sse V$ of size at most
  $r$, there exists a distribution $\mathcal{D}_S$ over subsets of $S$ and these distributions satisfy the following property: For any $Q \sse S$ and for every $A \sse
  Q$,
  \begin{equation}
    \label{eqn:cons}
    \mathcal{D}_Q(\{A\}) = \mathcal{D}_S(\{B : B \cap Q = A\}).
  \end{equation}
\end{theorem}

\ignore{
\begin{theorem} \label{thm:SAiffDist}
Consider a set $X$.  Say that for every subset $T$ of size at most
$k=2r+3$, there exists a distribution $\mathcal{D}_T$ over subsets of
$T$.  If for any $Q \sse T$ and for every $A \sse Q$ it holds that 
\begin{equation}
\label{eqn:cons}
\mathcal{D}_Q(\{A\}) = \mathcal{D}_T(\{B : B \cap Q = A\})
\end{equation}
then the vector $(y_{ij})|_{i,j \in X}$ such that $y_{ij} = \mathcal{D}_{\{i,j\}}(\{\{i\},\{j\}\})$ lies in the $r$-round \SA relaxation of the cut polytope.

Also, given a solution to the $r$-round \SA relaxation of the cut polytope, for every subset $T$ of size at most $r$ there exists a distribution $\mathcal{D}_T$ such that for any $Q \sse T$ and for every $A \sse Q$, \eqref{eqn:cons} holds.
\end{theorem}
}
  
  \begin{proof}
  Assume we have $(x,y) \in \mathbf{SA}_r(n)$.  For a set $S \sse V$ such that $|S| \leq r$, we define $\mathcal{D}_S$ such that $\mathcal{D}_S(T) = x(S,T)$.  By \eqref{eqn:sa_non_neg} and \eqref{eqn:sa_sum_to_1}, all $x(S,T) \geq 0$ and $\sum_{S \sse T} x(S,T) = 1$, so this is a well-defined distribution.  We now need to show that these distributions satisfy \eqref{eqn:cons}, i.e., we need to show that
  \[
  x(Q,A) = \sum_{\substack{B \sse S \\ B \cap Q = A}} x(S,B).
  \]
  This follows directly from \lref[Lemma]{lem:cons}:
  \begin{eqnarray*}
  \sum_{\substack{B \sse S \\ B \cap Q = A}} x(S,B) &=& \sum_{B' \sse S \setminus Q} x(Q \cup (S \setminus Q), B' \cup A) \\
  &=& x(Q,A).
  \end{eqnarray*}

  On the other hand, assume we have distributions $\mathcal{D}_S$ for all $S \sse V$ such that $|S| \leq r$.  Set $x(S,T) = \mathcal{D}_S(T)$.  The fact that $\mathcal{D}_S$ is a distribution implies that $\sum_{T \sse S} \mathcal{D}_S(T) = 1$ and the non-negativity constraints are satisfied.  Finally, \eqref{eqn:cons} implies that the consistency constraint
  \[
  \mathcal{D}_S(T) = \mathcal{D}_{S+u}(T) + \mathcal{D}_{S+u}(T+u)
  \]
  is satisfied.
  \end{proof}

\subsection{The Integrality Gap Instance}
\label{sec:SA-instance}

In this section, we show that the integrality gap of the LP for
\spcut remains $2-\eps$, even after polynomially many rounds of \SA.
Recall the construction from \lref[Section]{sec:hardness}: Given a
connected unweighted instance $H = (V, E_H)$ of \mc, it produces an
instance $G_{\ell}$ with vertices $V_\ell$ and supply edges $E_\ell$, such that
the sparsest cut in $G_\ell$ is related to the max cut in
$H$. We use the same construction here, but instead of starting with a
hard instance $H$ of \mc, we start with a \mc instance
exhibiting an integrality gap for $r$ rounds of \SA.

We first show that there exist ``local'' distributions satisfying the
conditions of \lref[Theorem]{thm:SAiffDist} for all subsets of the
vertices of $G_{\ell}$ of size at most $O(r)$  This implies the
existence of an $O(r)$-round \SA solution for the Sparsest Cut instance.
We then calculate the value of this fractional solution for $G_\ell$ and
relate it to the integral optimum.  This result is a natural extension
of those of~\cite{LR10,CSW10}, who showed a gap of $2$ for the basic LP
also on the fractal for $K_{2,n}$ using the fact that the complete
graph $K_n$ exhibits an integrality gap for the basic \mc LP.

\begin{lemma}
  \label{lem:SAexists}
  Consider an unweighted \mc instance $H$ for which we have an
  $r$-round \SA solution and $G_{\ell}$ constructed from $H$.
  For any $T \sse V_{\ell}$ of size at most $r$, there exists a
  distribution $\mathcal{D}_T$ over subsets of $T$ such that for any $Q
  \sse T$ and for every $A \sse Q$, \eqref{eqn:cons} holds.
\end{lemma}

\begin{proof}
  Let us first define these local distributions.  By
  \lref[Theorem]{thm:SAiffDist}, an $r$-round \SA solution for the
  \mc instance $H$ gives, for each subset $R \sse V_H$ of size at
  most $r$, a probability distribution $\mathcal{F}_R$ over subsets of
  $R$. Moreover, these local distributions $\mathcal{F}_R$ satisfy the
  consistency constraints~(\ref{eqn:cons}). We will use this set of
  distributions $\{ \F_R \}_{R \sse V_H: |R| \leq r}$ to define another
  set of distributions $\{ \D_T \}_{T \sse V_\ell: |T| \leq r}$ which also satisfy the consistency
  constraints \eqref{eqn:cons}.

%  Consider the instance $G_\ell$, and associate each
%  vertex in $V_\ell$ with a natural number in $\{0, 1, \ldots, \ell\}$,
%  which we will call its \emph{level}.  A vertex is a level-$i$ vertex
%  if it appears in $V_i$, but not in $V_{i-1}$. Given a level-$i$ vertex
%  $w$, we say that level-$(i-1)$ vertices $u$ and $v$ are \emph{parents}
%  of $w$ if $\{u,v\} \in E_{i-1}$ and $w$ is adjacent to $u$ and $v$ in
%  $G_i$; we then call $w$ a \emph{child} of $u$ and $v$.  For an edge
%  $\{u,v\} \in E_{i-1}$, we define $C_{uv}$ to be the set of $n$
%  children of $u$ and $v$. Similarly, for $w$ a child of $u$ and $v$, we
%  say that edges $\{u,w\}$ and $\{v,w\}$ are \emph{child edges} of
%  $\{u,v\}$. (Finally, this allows us to define ancestors and
%  descendants of edges and vertices.)

  Recall that $G_\ell$ is obtained by taking a copy of $G_1 = (V_1,
  E_1)$ and then replacing each edge $e \in G_1$ by a copy of instance
  $G_{\ell -1}$ (which we call $G_{\ell-1}^e$).  To define the
  distribution for $T \sse V_\ell$ with $|T| \leq r$, we first extend
  $T$ to a set $T'$ in the following way:
  \begin{OneLiners}
  \item[(i)] Set $T' = T$.
  \item[(ii)] Add $s$ and $t$ to $T'$.
  \item[(iii)] For $v \in V_1$, do the following:
  \begin{OneLiners}
  \item[(a)] If $T_{sv} := T \cap G_{\ell-1}^{sv}$ is not empty, add $v$ to
    $T'$ and recurse on $G^{sv}_{\ell-1}$.
  \item[(b)] If $T_{vt} := T \cap G_{\ell-1}^{vt}$ is not empty, add $v$ to
    $T'$ and recurse on $G^{vt}_{\ell-1}$.
  \end{OneLiners}
  \end{OneLiners}
 Let $T'_1$ be the vertices of $T'$ that are in $G_1$, i.e., $T'_1 = T'
  \cap V_1$, and let $T'_e$ be all vertices of $T'$ in the instance of $G_{\ell-1}$ corresponding to edge
  $e$ in $G_1$.

  The distribution $\D_T$ is given by the following recursive process that takes an instance $G_{\ell}$ and a set $T \in V_{\ell}$ and outputs a random subset $X \sse T$:
  \begin{OneLiners}
  \item[(i)] Put $s$ in $X$; $t$ will never be in $X$.
  \item[(ii)] Draw a subset $Y$ from $\F_{T'_1}$.  With probability 1/2, set $X \gets X \cup Y$.  With probability 1/2, set $X \gets X \cup (T'_1 \setminus Y)$.
  \item[(iii)] For each vertex $v \in V_1$, do the following:
  \begin{OneLiners}
  \item[(a)] If $v \in X$, set $X \gets X \cup T'_{sv}$ and recurse on $G^{vt}_{\ell-1}, T'_{vt}$.
  \item[(b)] If $v \notin X$, set $X \gets X \cup T'_{vt}$ and recurse on $G^{sv}_{\ell-1}, T'_{sv}$.
  \end{OneLiners}
  \item[(iv)] Note that $X \sse T'$, so output $X \cap T$.
  \end{OneLiners}
  Since we make draws from the probability distributions
  $\mathcal{F}_{T'_1}$, we should ensure that these distributions are
  well-defined. In particular, since $T'_1$ is a subset of $T'$ which
  could be larger than $|T|=r$, we need to show that the size of $T'_1$
  is at most $r$.  Indeed, if $T'_1$ contains some vertex $x$ not in
  $T$, this vertex has been added due to some $y \in T$ in
  $G^e_{\ell-1}$ and hence can be charged to $y \in T$.  This finishes
  the definition of the local distributions $\D_T$.

It remains to show that for any $A \sse Q \sse T$, the consistency
  condition~(\ref{eqn:cons}) holds. Let us introduce some notation: For
  sets $Y \sse X$, let $\Pr_X(\pick Y)$ indicate the probability that
  $Y$ is chosen from the distribution $\D_X$. %  For disjoint sets $X, Z$
%   and for $Y \sse X$ and $Z \sse W$, let $\Pr_{X | W \sim Z}(Y)$
%   indicate the probability that $Y$ is chosen from $\D_X$ given that $W$
%   has already been drawn from $\D_Z$.
  Using this notation,
  (\ref{eqn:cons}) is the same as showing that for $T \sse V_{\ell}$ of
  size at most $r$ and any $A \sse Q \sse T$,
  \begin{equation}
    \label{eqn:cons_pr}
    \Pr_{Q}(\pick A) = \sum_{B \sse T : B \cap Q = A} \Pr_{T}(\pick B).
  \end{equation}
  It is easy to see that it suffices to prove this for the case where $T
  = Q + q$ for some $q \in V_\ell \setminus Q$. In this case, things simplify to
  showing that for every $A \sse Q$,
  \begin{equation}
    \label{eqn:cons_pr2}
    \Pr_{Q}(\pick A) = \Pr_{T}(\pick A) + \Pr_T(\pick (A + q)).
  \end{equation}

To prove \eqref{eqn:cons_pr2}, it then suffices to show
\begin{equation}
\label{eqn:conswts}
\Pr_{Q'}(\pick A) = \Pr_{T'}(\pick A) + \Pr_{T'}(\pick(A + q))
\end{equation}
for any $A \sse Q'$.
Recall that even though $|T \setminus Q| = 1$, the extended sets $T'$ and $Q'$ may differ in more vertices.

We proceed by induction on $\ell$.  In the base case, consider $G_1$.  If $q$ is $s$ or $t$, this is trivial.  Otherwise, the claim holds by the consistency properties of the $\{\F\}$ distributions.  In the inductive case, we assume \eqref{eqn:conswts} holds for $G_{\ell-1}$ and want to show that this claim holds for $G_{\ell}$.    There are two subcases: $q \in V_1$ and $q \notin V_1$.  If $q \in V_1$, the sets $Q'$ and $T'$ only differ on vertices in $V_1$ and the claim holds by the consistency of the $\{\F\}$ distributions.

Otherwise, $q \notin V_1$.  Let $Q'_1 = Q' \cap V_1$ and $A_1 = A \cap V_1$.  We denote by $V^e_{\ell-1}$ the vertices of $G^e_{\ell-1}$.  Let $A_e$ be $A \cap V^e_{\ell-1}$ and let $Q'_e = Q' \cap V^e_{\ell-1}$ for any edge $e$ in $G_1$.  Because our selection process is independent on each of the $V^e_{\ell-1}$'s given its choice in $V_1$, we have that
\[
\Pr_{Q'}(\pick A) = \Pr_{Q'_1}(\pick A_1) \prod_{e \in G_1} \Pr_{Q'_e}(\pick A_e~|~\pick A_1)
\]
and
\[
\Pr_{T'}(\pick A) = \Pr_{T'_1}(\pick A_1) \prod_{e \in G_1} \Pr_{T'_e}(\pick A_e~|~\pick A_1).
\]
Since $q \notin V_1$, $Q'_1 = T'_1$.  Also, since $T = Q+q$, observe that $Q'_e$ and $T'_e$ must be the same for all edges $e$ of $G_1$ except for one.  Call this edge $e^{\star}$.  These two facts imply that
\begin{equation}
\label{eqn:picka}
\Pr_{T'}(\pick A) = \Pr_{Q'_1}(\pick A_1)  \Pr_{T'_{e^{\star}}}(\pick A_{e^{\star}}~|~\pick A_1) \prod_{e \neq e^{\star}} \Pr_{Q'_e}(\pick A_e~|~\pick A_1).
\end{equation}
Similarly, we know that
\begin{equation}
\label{eqn:pickaq}
\Pr_{T'}(\pick (A+q)) = \Pr_{Q'_1}(\pick A_1)  \Pr_{T'_{e^{\star}}}(\pick (A_{e^{\star}}+q)~|~\pick A_1) \prod_{e \neq e^{\star}} \Pr_{Q'_e}(\pick A_e~|~\pick A_1).
\end{equation}
By the inductive assumption, we know that
\[
\Pr_{Q'_{e^{\star}}}(\pick A_{e^{\star}}) = \Pr_{T'_{e^{\star}}}(\pick A_{e^{\star}}~|~\pick A_1)+ \Pr_{T'_{e^{\star}}}(\pick (A_{e^{\star}}+q)~|~\pick A_1).
\]
Adding \eqref{eqn:picka} and \eqref{eqn:pickaq} therefore gives us
\[
\Pr_{T'}(\pick A) + \Pr_{T'}(\pick (A+q)) = \Pr_{Q'_1}(\pick A_1) \prod_{e \in G_1} \Pr_{Q'_e}(\pick A_e~|~\pick A_1) = \Pr_{Q'}(\pick A).
\]

\end{proof}

Now, applying \lref[Theorem]{thm:SAiffDist} to the \mc instance from \ifstoc \\ \fi 
\lref[Lemma]{lem:SAexists}, we get an $r$-round \SA solution $y$ for
Sparsest Cut. Recall that we set $y_e =
\mathcal{D}_{\{i,j\}}(\{\{i\},\{j\}\})$, which is the probability over
the distribution $\mathcal{D}_{\{i,j\}}$ that exactly one of the
endpoints is chosen. We analyze the Sparsest Cut objective function
value of this fractional solution $y$ next. 

\begin{lemma}
  \label{lem:SAval}
  Let the value of the $r$-round \SA relaxation solution $z$ for the
  \mc instance on $H$ implied by the distributions $\{\F_S\}$ be
  $\sum_{e \in E_H} z_e = c\,m$.  Then the sparsity of the $r$-round
  solution $y$ is $\frac{1}{\ell c}$.
\end{lemma}

\begin{proof}
  Let $\capa_{ij}$ be the capacity of edge $\{i,j\}$ and $\dem_{ij}$ be
  the demand on edge $\{i,j\}$ in the instance $G_\ell$.  Let
  $E^c_{\ell}$ be its capacity edges and $E^D_{\ell}$ its demand edges.

  We begin by proving that $\sum_{e \in E^c_{\ell}} \capa_e y_e = 1$.
  First, we claim that $y_e = 2^{-\ell}$ for any capacity edge $e$.  
By the symmetry introduced in Step~ii(c) of the definition of $\D_T$, a capacity edge of $G_1$ is cut with probability $1/2$, so $y_e = 1/2$ for any such edge $e$.  Using this fact, a simple induction then suffices.
% Full proof here: am just writing a synopsis above.
% Let $e = \{i,j\}$.  In the base case, we consider $G_1$.
%   By construction, a level 1 vertex is separated from $s$ with
%   probability $\frac{1}{2}$ is separated from $t$ with probability
%   $\frac{1}{2}$.  Now assume that an edge of $G_{\ell-1}$ is cut with
%   probability $2^{-\ell+1}$.  Any level $\ell$ vertex will end up on
%   either side of the cut with probability $1/2$.  For an edge of
%   $G_{\ell}$ to be cut, the corresponding edge of $G_{\ell-1}$ needs to
%   be cut and the level $\ell$ vertex of the edge needs to end up on the
%   opposite side as the other endpoint of the edge.  The probability that
%   both of these events happen is $\frac{1}{2}2^{-\ell+1} = 2^{-\ell}$.
  Second, we claim that $\sum_{e \in E^c_{\ell}} \capa_e = 2^{\ell}$.  Since $\sum_{e \in E^c_1} \capa_e = 2$, a simple induction again gives us the desired result.
%  Secondly, we claim that $\sum_{e \in E^c_{\ell}} \capa_e = 2^{\ell}$.
%  Indeed, when we replace an edge $\{u,v\}$ by its $2n$ child edges, the
%  capacities $\capa_{ui}, \capa_{vi}$ are $ \capa_{uv} \cdot (\deg_H(i)/2m)$;
%  summing over all $i \in C_{uv}$ gives us $2\capa_{uv}$, hence doubling the
%  capacity. Since the capacity of $G_1$ is $2$, a simple induction gives
%  the total capacity of $G_\ell$ being $2^\ell$.
  % Full proof here: am just writing a synopsis above.
%   again by
%   induction on $\ell$.  In the base case, we again consider $G_1$.  For
%   each vertex of $H$, we have two edges that each have capacity
%   $\mathrm{deg}_H(i)/2$.  Since $2\sum_i \mathrm{deg}_H(i)/2 = 2m$, the
%   base case holds.  Assume the sum of the capacities of all edges of
%   $G_{\ell-1}$ is $2^{\ell-1}m$.  Let $\{u,v\}$ be some edge of
%   $G_{\ell-1}$ with capacity $c_{uv}$.  Then in $G_{\ell}$, we replace
%   $\{u,v\}$ with $\{u,v\} \times [n]$ such that the capacities $c_{ui}$
%   and $c_{vi}$ of $\{u,i\}$ and $\{v,i\}$ are equal to $c_{uv} \cdot
%   \frac{\mathrm{deg}_H(i)}{2m}$.  So, the total capacity of the edges
%   replacing $c_{uv}$ is $2\sum_i c_{uv} \cdot
%   \frac{\mathrm{deg}_H(i)}{2m} = 2c_{uv}$.  The capacity of edges in
%   $G_{\ell-1}$ is doubled in $G_{\ell}$, giving us the desired result.
%   Therefore, $\sum_{e \in E^c_{\ell}} c_e y_e = m$.
  Combining these facts, we get $\sum_{e \in E^c_{\ell}} \capa_e y_e = 1$.

  Next, we show that $\sum_{e \in E^D_{\ell}} \dem_e y_e = \ell
  c$. The proof is again by induction on $\ell$. The base case is
  $G_1$, where the value is $c$ because $y_e = z_e$ for $e \in E_H$ and all demands are $1/m$. For the induction step, assume $\sum_{e
    \in E^D_{\ell-1}} \dem_e y_e = (\ell-1)c$. For each $e \in E^c_1$, the corresponding $G^e_{\ell-1}$ contributes $\frac{1}{2} (\ell-1)c \cdot \capa_e$ by the inductive hypothesis and the fact that $y_{e'} = \frac{1}{2}y_{e''}$ for all $e' \in E^c_{\ell}$ and $e'' \in E^c_{\ell-1}$.  Summing over all $e \in E^c_1$ gives $(\ell-1)c$.  By the base case, level $\ell$ demands contribute $c$, giving a total of $\ell c$ demand cut.
\end{proof}

Given the above construction, we can now prove the main theorem of this
section, which allows us to convert \SA integrality gaps for \mc to
\SA integrality gaps for Sparsest Cut.

\begin{theorem}
  \label{thm:translate}
  Given an unweighted \mc instance $H$ on $n$ nodes and $m$ edges with a max cut
  of size $sm$ and an $r$-round \SA solution of value at least $cm$, for
  any constant $\eps > 0$ there exists a Sparsest Cut instance $G$ with
  an $r$-round \SA integrality gap of $(c/s) - \eps$, such
  that the size of $G$ is $n^{O(c/s^2 \eps)}$.
\end{theorem}

\begin{proof}
  We will set $G$ to be the graph $G_{\ell}$ constructed from base
  instance $H$ for some value of $\ell$ to be chosen later. By Lemma
  \ref{lem:SAexists}, there exists an $r$-round \SA solution
  $y$ whose sparsity is $\frac{1}{\ell c}$ by \lref[Lemma]{lem:SAval}. By calculations as in
  Section~\ref{sec:endgame-NP}, the
  actual sparsest cut value for $G_{\ell}$ is at least $\frac{1}{1+(\ell-1)s}$. 
  This gives us an $r$-round \SA integrality gap of
  $\frac{\ell c}{1+(\ell-1) s} \geq \frac{c}{s} - \eps$ for $\ell = \frac{c}{s^2 \eps}
$.
\end{proof}

Plugging in the \SA integrality gap instance for \mc due to Charikar
et al.~\cite{CMM09} gives us the following corollary:
\begin{corollary}
  For every $\eps > 0$, there exists $\gamma > 0$ such that the
  integrality gap of the Sparsest Cut relaxation is $2 - \eps$ even
  after $n^\gamma$ rounds of \SA.
\end{corollary}

\begin{proof}
  From \cite[Theorem 5.3]{CMM09}, we have that for any $\eps' > 0$,
  there exists $\gamma' > 0$ and an unweighted \mc instance with $M$
  edges and $N$ nodes such that the optimal integral solution is at
  most $M(1/2 + \eps'/6)$ and the LP value after $N^{\gamma'}$ rounds
  of \SA is at least $N(1 - \eps'/6)$.  From this, we get $c/s > 2 - \eps'$. Now
  \lref[Theorem]{thm:translate} allows us, for any $\eps'' > 0$, to obtain
  a Sparsest Cut instance with integrality gap $2 - \eps'-\eps''$ after
  $N^{\gamma'}$ rounds of \SA. Setting $\eps' = \eps'' =
  \eps/2$ gives us the integrality gap of $2 - \eps$.  Moreover, the
  size of the new instance is $n = N^{O(c/s^2 \eps)}$, so setting $\gamma =
  O(s^2 \eps \gamma'/c)$ completes the proof.
\end{proof}

%%% Local Variables: 
%%% mode: latex
%%% TeX-master: "sp-cut"
%%% End: 

\fi

\section{Conclusions}

We show how to use the \SA hierarchy to get a factor-$2$ approximation
for the \nusc problem on tree\-width-$k$ graphs in time $n^{O(k)}$. (This
also gives $2^{\tilde{O}(\sqrt{n})}$-time $2$-approx\-imation algorithms
for Sparsest Cut on minor-free graphs.) We also show that the \nusc
problem is as hard as the \mc problem, even for treewidth-$2$ graphs,
which gives us the best NP-hardness known (even for the unconstrained
problem). Assuming the UGC, this gives a hardness of $1/0.878 - \eps$ 
for these series-parallel graphs. For graphs of large constant
treewidth, we show a Unique Games hardness of $2 - \eps$, which matches
our algorithm. Finally, we demonstrate an integrality gap of $2 - \eps$
for \SA relaxations after a polynomial number of rounds, even for
treewidth-$2$ graphs.

Many research directions remain open. Among them are getting better
hardness results for \nusc, both for restricted graph classes and
for the general problem, getting poly\-nomial-time $O(1)$-approximation
algorithms for planar or minor-closed families (using LP/SDP hierarchies
or otherwise), and making progress on the embeddability conjecture
from~\cite{GNRS99}.

% Several questions remain: a natural one is whether we can use stronger
% convex relaxations (such as those given by the Lasserre hierarchy) to
% beat the factor $2$ algorithms, say even for treewidth-$2$ graphs?
% Using the rounding scheme of Raghavendra and Tan~\cite{RT12} (see
% also~\cite{ABG13}

\subsection*{Acknowledgments}

We thank Venkat Guruswami, Ryan O'Donnell, Robi Krauth\-gamer, Prasad
Raghavendra, and Eden Chlamt\'a\v{c} for useful discussions, Claire
Mathieu for the blog post that inspired us to revisit this problem, Aravindan Vijayaraghavan for pointing out the
\textsc{Apx}-hardness in~\cite{CK06}, and the anonymous referees for several helpful comments. Part of this work was done when
the first-named author was at Microsoft Research SVC and Columbia University; he thanks them for their generous
hospitality.

\ifstoc
 \bibliographystyle{abbrv}
 \bibliography{sp-cut}
\else
 \bibliographystyle{alpha}
 {\small \bibliography{sp-cut}}

\newcommand{\etalchar}[1]{$^{#1}$}
\begin{thebibliography}{KKMO07}

\bibitem[ACP87]{ACP87}
S.~Arnborg, D.~Corneil, and A.~Proskurowski.
\newblock Complexity of finding embeddings in a k-tree.
\newblock {\em SIAM Journal on Algebraic Discrete Methods}, 8(2):277--284,
  1987.

\bibitem[ALN08]{ALN05}
Sanjeev Arora, James~R. Lee, and Assaf Naor.
\newblock Euclidean distortion and the sparsest cut.
\newblock {\em J. Amer. Math. Soc.}, 21(1):1--21, 2008.

\bibitem[Ami10]{Amir10}
Eyal Amir.
\newblock Approximation algorithms for treewidth.
\newblock {\em Algorithmica}, 56(4):448--479, 2010.

\bibitem[AMS11]{AMS07}
Christoph Amb{\"u}hl, Monaldo Mastrolilli, and Ola Svensson.
\newblock Inapproximability results for maximum edge biclique, minimum linear
  arrangement, and sparsest cut.
\newblock {\em SIAM J. Comput.}, 40(2):567--596, 2011.

\bibitem[ARV09]{ARV04}
Sanjeev Arora, Satish Rao, and Umesh Vazirani.
\newblock Expander flows, geometric embeddings and graph partitioning.
\newblock {\em J. ACM}, 56(2):Art. 5, 37, 2009.

\bibitem[BC05]{BC03}
Bo~Brinkman and Moses Charikar.
\newblock On the impossibility of dimension reduction in {$l_1$}.
\newblock {\em J. ACM}, 52(5):766--788, 2005.

\bibitem[BO04]{BO04}
Daniel Bienstock and Nuri Ozbay.
\newblock Tree-width and the {S}herali-{A}dams operator.
\newblock {\em Discrete Optim.}, 1(1):13--21, 2004.

\bibitem[Bod89]{Bod89}
Hans Bodlaender.
\newblock {NC}-algorithms for graphs with small treewidth.
\newblock In {\em Graph-Theoretic Concepts in Computer Science}, volume 344 of
  {\em LNCS}, pages 1--10. 1989.

\bibitem[Bod96]{Bod96}
H.~Bodlaender.
\newblock A linear-time algorithm for finding tree-decompositions of small
  treewidth.
\newblock {\em SIAM Journal on Computing}, 25(6):1305--1317, 1996.

\bibitem[Bod98]{Bod98}
Hans~L. Bodlaender.
\newblock A partial $k$-arboretum of graphs with bounded treewidth.
\newblock {\em Theoretical Computer Science}, 209:1--45, 1998.

\bibitem[CGN{\etalchar{+}}06]{CGNRS01}
Chandra Chekuri, Anupam Gupta, Ilan Newman, Yuri Rabinovich, and Alistair
  Sinclair.
\newblock Embedding {$k$}-outerplanar graphs into {$\ell_1$}.
\newblock {\em SIAM J. Discrete Math.}, 20(1):119--136, 2006.

\bibitem[CJLV08]{CJLV08}
Amit Chakrabarti, Alexander Jaffe, James~R. Lee, and Justin Vincent.
\newblock Embeddings of topological graphs: Lossy invariants, linearization,
  and 2-sums.
\newblock In {\em FOCS}, pages 761--770, 2008.

\bibitem[CK09]{CK06}
Julia Chuzhoy and Sanjeev Khanna.
\newblock Polynomial flow-cut gaps and hardness of directed cut problems.
\newblock {\em J. ACM}, 56(2):Art. 6, 28, 2009.

\bibitem[CKK{\etalchar{+}}06]{CKKRS05}
Shuchi Chawla, Robert Krauthgamer, Ravi Kumar, Yuval Rabani, and D.~Sivakumar.
\newblock On the hardness of approximating multicut and sparsest-cut.
\newblock {\em Comput. Complexity}, 15(2):94--114, 2006.

\bibitem[CKN09]{CKN09}
Jeff Cheeger, Bruce Kleiner, and Assaf Naor.
\newblock A {$(\log n)^{\Omega(1)}$} integrality gap for the sparsest cut
  {SDP}.
\newblock In {\em FOCS}, pages 555--564. 2009.

\bibitem[CKN11]{CKN11}
Jeff Cheeger, Bruce Kleiner, and Assaf Naor.
\newblock Compression bounds for {L}ipschitz maps from the {H}eisenberg group
  to {$L_1$}.
\newblock {\em Acta Math.}, 207(2):291--373, 2011.

\bibitem[CKR10]{CKR10}
Eden Chlamt\'a\v{c}, Robert Krauthgamer, and Prasad Raghavendra.
\newblock Approximating sparsest cut in graphs of bounded treewidth.
\newblock In {\em APPROX}, volume 6302 of {\em LNCS}, pages 124--137. 2010.

\bibitem[CMM09]{CMM09}
Moses Charikar, Konstantin Makarychev, and Yury Makarychev.
\newblock Integrality gaps for {S}herali-{A}dams relaxations.
\newblock In {\em S{TOC}}, pages 283--292. 2009.

\bibitem[CST01]{CST96}
Pierluigi Crescenzi, Riccardo Silvestri, and Luca Trevisan.
\newblock On weighted vs unweighted versions of combinatorial optimization
  problems.
\newblock {\em Inform. and Comput.}, 167(1):10--26, 2001.

\bibitem[CSW10]{CSW10}
Chandra Chekuri, F.~Bruce Shepherd, and Christophe Weibel.
\newblock Flow-cut gaps for integer and fractional multiflows.
\newblock In {\em SODA}, pages 1198--1208, 2010.

\bibitem[Die00]{diestel}
Reinhard Diestel.
\newblock {\em Graph theory}, volume 173 of {\em Graduate Texts in
  Mathematics}.
\newblock Springer-Verlag, New York, 2000.

\bibitem[Fri98]{Freidgut98}
Ehud Friedgut.
\newblock Boolean functions with low average sensitivity depend on few
  coordinates.
\newblock {\em Combinatorica}, 18(1):27--35, 1998.

\bibitem[GNRS04]{GNRS99}
Anupam Gupta, Ilan Newman, Yuri Rabinovich, and Alistair Sinclair.
\newblock Cuts, trees and {$\ell_1$}-embeddings of graphs.
\newblock {\em Combinatorica}, 24(2):233--269, 2004.

\bibitem[GS13]{GS12}
Venkatesan Guruswami and Ali~Kemal Sinop.
\newblock Certifying graph expansion and non-uniform sparsity via generalized
  spectra.
\newblock In {\em SODA}, 2013.

\bibitem[GSZ12]{GSZ12}
Venkatesan Guruswami, Ali~Kemal Sinop, and Yuan Zhou.
\newblock Constant factor lasserre integrality gaps for graph partitioning
  problems.
\newblock {\em CoRR}, abs/1202.6071, 2012.

\bibitem[H{\aa}s01]{Hastad01}
Johan H{\aa}stad.
\newblock Some optimal inapproximability results.
\newblock {\em J. ACM}, 48(4):798--859, 2001.

\bibitem[KK97]{KK94}
David~R. Karger and Daphne Koller.
\newblock ({D}e)randomized construction of small sample spaces in {NC}.
\newblock {\em J. Comput. System Sci.}, 55(3):402--413, 1997.

\bibitem[KKMO07]{KKMO}
Subhash Khot, Guy Kindler, Elchanan Mossel, and Ryan O'Donnell.
\newblock Optimal inapproximability results for {MAX}-{CUT} and other
  2-variable {CSP}s?
\newblock {\em SIAM J. Comput.}, 37(1):319--357, 2007.

\bibitem[KPR93]{KPR93}
Philip Klein, Serge~A. Plotkin, and Satish~B. Rao.
\newblock Excluded minors, network decomposition, and multicommodity flow.
\newblock In {\em STOC}, pages 682--690, 1993.

\bibitem[KR09]{KR09}
Robert Krauthgamer and Yuval Rabani.
\newblock Improved lower bounds for embeddings into {$L_1$}.
\newblock {\em SIAM J. Comput.}, 38(6):2487--2498, 2009.

\bibitem[KS09]{KS09}
Subhash Khot and Rishi Saket.
\newblock S{DP} integrality gaps with local {$\ell_1$}-embeddability.
\newblock In {\em {FOCS} 09}, pages 565--574. IEEE Computer Soc., Los Alamitos,
  CA, 2009.

\bibitem[KV05]{KV05}
Subhash Khot and Nisheeth~K. Vishnoi.
\newblock The unique games conjecture, integrality gap for cut problems and
  embeddability of negative type metrics into l$_{\mbox{1}}$.
\newblock In {\em FOCS}, pages 53--62, 2005.

\bibitem[LLR95]{LLR95}
Nathan Linial, Eran London, and Yuri Rabinovich.
\newblock The geometry of graphs and some of its algorithmic applications.
\newblock {\em Combinatorica}, 15(2):215--245, 1995.

\bibitem[LN04]{LN03}
J.~R. Lee and A.~Naor.
\newblock Embedding the diamond graph in {$L_p$} and dimension reduction in
  {$L_1$}.
\newblock {\em Geom. Funct. Anal.}, 14(4):745--747, 2004.

\bibitem[LN06]{LeeNaor06}
James~R. Lee and Assaf Naor.
\newblock $l_p$ metrics on the {Heisenberg} group and the {Goemans}-{Linial}
  conjecture.
\newblock In {\em FOCS}, pages 99--108, 2006.

\bibitem[LR10]{LR10}
James~R. Lee and Prasad Raghavendra.
\newblock Coarse differentiation and multi-flows in planar graphs.
\newblock {\em Discrete Comput. Geom.}, 43(2):346--362, 2010.

\bibitem[LS09]{LS09}
James~R. Lee and Anastasios Sidiropoulos.
\newblock On the geometry of graphs with a forbidden minor.
\newblock In {\em S{TOC}}, pages 245--254. 2009.

\bibitem[LS11]{LS11}
James~R. Lee and Anastasios Sidiropoulos.
\newblock Near-optimal distortion bounds for embedding doubling spaces into
  {$L_1$} [extended abstract].
\newblock In {\em S{TOC}}, pages 765--772. 2011.

\bibitem[MM09]{MM09}
Avner Magen and Mohammad Moharrami.
\newblock Robust algorithms for max independent set on minor-free graphs based
  on the {S}herali-{A}dams hierarchy.
\newblock In {\em APPROX}, volume 5687 of {\em LNCS}, pages 258--271. Springer,
  Berlin, 2009.

\bibitem[MOO10]{MOO}
Elchanan Mossel, Ryan O'Donnell, and Krzysztof Oleszkiewicz.
\newblock Noise stability of functions with low influences: invariance and
  optimality.
\newblock {\em Ann. of Math. (2)}, 171(1):295--341, 2010.

\bibitem[MS90]{MS90}
David~W. Matula and Farhad Shahrokhi.
\newblock Sparsest cuts and bottlenecks in graphs.
\newblock {\em Discrete Appl. Math.}, 27(1-2):113--123, 1990.

\bibitem[NR03]{NR02}
Ilan Newman and Yuri Rabinovich.
\newblock A lower bound on the distortion of embedding planar metrics into
  {E}uclidean space.
\newblock {\em Discrete Comput. Geom.}, 29(1):77--81, 2003.

\bibitem[OS81]{OS81}
Haruko Okamura and P.~D. Seymour.
\newblock Multicommodity flows in planar graphs.
\newblock {\em J. Combin. Theory Ser. B}, 31(1):75--81, 1981.

\bibitem[Rab03]{Yuri03}
Yuri Rabinovich.
\newblock On average distortion of embedding metrics into the line and into
  $\ell_1$.
\newblock In {\em STOC}, pages 456--462, 2003.

\bibitem[Rao99]{Rao99}
Satish Rao.
\newblock Small distortion and volume preserving embeddings for planar and
  {E}uclidean metrics.
\newblock In {\em SOCG}, pages 300--306, 1999.

\bibitem[Ree92]{Reed92}
Bruce~A. Reed.
\newblock Finding approximate separators and computing tree width quickly.
\newblock In {\em STOC}, pages 221--228, 1992.

\bibitem[Reg12]{Regev12}
Oded Regev.
\newblock Entropy-based bounds on dimension reduction in $l_1$.
\newblock {\em Israel Journal of Mathematics}, 2012.
\newblock arXiv:1108.1283.

\bibitem[RS86]{RS-II}
Neil Robertson and P.~D. Seymour.
\newblock Graph minors. {II}. {A}lgorithmic aspects of tree-width.
\newblock {\em J. Algorithms}, 7(3):309--322, 1986.

\bibitem[RS09]{RS09}
Prasad Raghavendra and David Steurer.
\newblock Integrality gaps for strong {SDP} relaxations of {U}nique {G}ames.
\newblock In {\em FOCS}, pages 575--585. 2009.

\bibitem[TSSW00]{TSSW}
Luca Trevisan, Gregory~B. Sorkin, Madhu Sudan, and David~P. Williamson.
\newblock Gadgets, approximation, and linear programming.
\newblock {\em SIAM J. Comput.}, 29(6):2074--2097, 2000.

\bibitem[WJ04]{WJ04}
Martin~J. Wainwright and Michael~I. Jordan.
\newblock Treewidth-based conditions for exactness of the sherali-adams and
  lasserre relaxations.
\newblock Technical Report 671, UC Berkeley, September 2004.

\end{thebibliography}
\fi

\end{document}